\documentclass[12pt]{article}

\usepackage{amsthm,amsmath,amssymb}
\usepackage{comment} 
\usepackage[sort]{cite}

\usepackage{graphicx}

\usepackage[colorlinks=true,citecolor=black,linkcolor=black,urlcolor=blue]{hyperref}

\theoremstyle{plain}
\newtheorem{theorem}{Theorem}[section]
\newtheorem{lemma}[theorem]{Lemma}
\newtheorem{corollary}[theorem]{Corollary}
\newtheorem{observation}[theorem]{Observation}
\newtheorem{protocol}{Protocol}

\theoremstyle{definition}
\newtheorem{definition}[theorem]{Definition}
\newtheorem{example}[theorem]{Example}

\theoremstyle{remark}
\newtheorem{remark}[theorem]{Remark}

\DeclareMathOperator{\TD}{TD}
\DeclareMathOperator{\OA}{OA}
\newcommand{\G}{\ensuremath{\mathcal{G}}} 
\DeclareMathOperator{\STS}{STS}
\DeclareMathOperator{\SkS}{S}
\DeclareMathOperator{\AG}{AG}

\newcommand{\A}{\ensuremath{\mathcal{A}}}
\newcommand{\B}{\ensuremath{\mathcal{B}}}
\newcommand{\C}{\ensuremath{\mathcal{C}}}
\newcommand{\D}{\ensuremath{\mathfrak{G}}}

\title{\bf Additional Constructions to Solve the Generalized Russian Cards Problem using Combinatorial Designs\footnote{Much of this work appears in the PhD thesis of
the first author~\cite{SwansonPhD_13}.}}

\author{Colleen M. Swanson \thanks{This work was supported in part by the TerraSwarm Research Center, one of six centers supported by the STARnet phase of the Focus Center Research Program (FCRP) a Semiconductor Research Corporation program sponsored by MARCO and DARPA.}\\
\small Computer Science \& Engineering Division\\[-0.8ex]
\small University of Michigan\\[-0.8ex] 
\small Ann Arbor, MI 48109, USA.\\
\small\tt cmswnsn@umich.edu\\
\and
Douglas~R.~Stinson\thanks{Research supported by NSERC grant 203114-11}\\
\small David R. Cheriton School of Computer Science\\[-0.8ex]
\small University of Waterloo\\[-0.8ex]
\small Waterloo, Ontario, N2L 3G1 Canada\\
\small\tt dstinson@uwaterloo.ca
}

\begin{document}

\maketitle


\begin{abstract}
In the generalized Russian cards problem, we have a card deck $X$ of $n$ cards and three participants, Alice, Bob, and Cathy, dealt $a$, $b$, and $c$ cards, respectively. Once the cards are dealt, Alice and Bob wish to privately communicate their hands to each other via public announcements, without the advantage of a shared secret or public key infrastructure. Cathy, for her part, should remain ignorant of all but her own cards after Alice and Bob have made their announcements. Notions for Cathy's ignorance in the literature range from Cathy not learning the fate of any individual card with certainty (\emph{weak $1$-security}) to not gaining any probabilistic advantage in guessing the fate of some set of $\delta$ cards (\emph{perfect $\delta$-security}). As we demonstrate in this work, the generalized Russian cards problem has close ties to the field of combinatorial designs, on which we rely heavily, particularly for perfect security notions. Our main result establishes an equivalence between perfectly $\delta$-secure strategies and $(c+\delta)$-designs on $n$ points with block size $a$, when announcements are chosen uniformly at random from the set of possible announcements. We also provide construction methods and example solutions, including a construction that yields perfect $1$-security against Cathy when $c=2$. Drawing on our equivalence results, we are able to use a known combinatorial design to construct a strategy with $a=8$, $b=13$, and $c=3$ that is perfectly $2$-secure. Finally, we consider a variant of the problem that yields solutions that are easy to construct and optimal with respect to both the number of announcements and level of security achieved. Moreover, this is the first method obtaining weak $\delta$-security that allows Alice to hold an arbitrary number of cards and Cathy to hold a set of $c = \lfloor \frac{a-\delta}{2} \rfloor$ cards. Alternatively, the construction yields solutions for arbitrary $\delta$, $c$ and any $a \geq \delta + 2c$.

\end{abstract}

\section{Introduction}
In the generalized Russian cards problem, we have a card deck $X$ and three participants, Alice, Bob, and Cathy. Once the cards are dealt, Alice and Bob wish to privately communicate their hands to each other via public announcements, without the advantage of a shared secret or public key infrastructure. Here we focus on protocols of length two, which allows us to consider only Alice's announcement. That is, Alice should make an \emph{informative} announcement, so that Bob learns the card deal. Bob, after hearing Alice's informative announcement, can always announce Cathy's hand. Cathy, for her part, should remain ignorant of all but her own cards after Alice and Bob have made their announcements. 

Notions for Cathy's ignorance in the literature range from Cathy not learning the fate of any individual card with certainty (\emph{weak $1$-security}) to not gaining any probabilistic advantage in guessing the fate of some set of $\delta$ cards (\emph{perfect $\delta$-security}), where here we are referring to cards not already held by Cathy. As we discuss in this work, the generalized Russian cards problem has close ties to the field of combinatorial designs, on which we rely heavily, particularly for perfect security notions.

If a scheme satisfies weak $1$-security, Cathy should not be able to say whether a given card is held by Alice or Bob (unless she holds the card herself). If a scheme satisfies perfect $1$-security, each card is equally likely to be held by Alice. When Alice's strategy is \emph{equitable} (in the sense that Alice picks uniformly at random from some set of possible announcements), we show an equivalence between perfectly secure strategies and sets of $2$-designs on $n$ points with block size $a$.
	
	Generalizing these notions of weak and perfect security, which focus on the probability that individual cards are held by Alice, we consider instead the probability that a given set of $\delta$ cards is held by Alice. If the probability distribution is uniform across $\delta$-sets, we say the scheme satisfies perfect $\delta$-security, and if the distribution is not uniform (but positive for every possible $\delta$-set), then we  have weak $\delta$-security. We consider equitable strategies and show an equivalence between perfectly $\delta$-secure strategies and $(c+\delta)$-designs on $n$ points with block size $a$. For equitable, informative, and perfectly $(a-c-1)$-secure strategies, we show $c=1$ and demonstrate an equivalence between these strategies and \emph{Steiner systems} $\SkS(a-1,a,n)$, a result first shown in Swanson and Stinson~\cite{cardsSwansonS12}, albeit with a much more complicated proof than we present here. 
	
	Building on results in Swanson and Stinson~\cite{cardsSwansonS12}, we show how to use a $t$-$(n,a,1)$-design to construct equitable $(a,b,c)$-strategies that are informative for Bob and perfectly $(t-c)$-secure against Cathy for any choice of $c$ satisfying $c \leq \min \{t-1,a-t\}$. In particular, this indicates that if an appropriate $t$-design exists, it is possible to achieve perfect security for deals where Cathy holds more than one card.  We present an example construction, based on \emph{inversive planes}, for $(q+1, q^2-q-2, 2)$-strategies which are perfectly 1-secure against Cathy and informative for Bob, where $q$ is a prime power. This example, first given in Swanson~\cite{SwansonPhD_13}, is among the first strategies presented in the literature that is informative for Bob and achieves perfect $1$-security against Cathy for $c > 1$. This example was found independently from the work of Cord\'{o}n-Franco et al.~\cite{CDFS13}, discussed later, which presents a protocol that for certain parameters achieves perfect 1-security against Cathy for $c=2$. Finally, our results allow us to draw on a known combinatorial designs in order to realize a perfectly $2$-secure $(8,13,3)$-strategy, which shows that it is possible, at least for some deals, to achieve perfect security for $c > 2$.

	Finally, we discuss a variation on the generalized Russian cards problem, where the card deck is first split into $a$ piles, and Alice and Cathy's hands consist of at most one card from each pile, with Bob receiving the remaining cards. This variant admits a nice solution using \emph{transversal designs} with $\lambda =1$ that achieves weak $(a-2c)$-security. In particular, this solution is easy to construct and is optimal with respect to both the number of announcements and level of security achieved. Moreover, this is the first method obtaining weak $\delta$-security that allows Alice to hold an arbitrary number of cards and Cathy to hold a set of $c = \lfloor \frac{a-\delta}{2} \rfloor$ cards. Alternatively, the construction yields solutions for arbitrary $\delta$, $c$ and any $a \geq \delta + 2c$.

\section{Paper outline} 
After reviewing basic results from combinatorial designs in Section~\ref{sec:combinatorial_designs}, we review the basic framework for the generalized Russian cards problem and establish relevant notation in Section~\ref{subsec: Preliminary Notation and Examples}. In Section~\ref{sec:informative_strategies}, we study and define the notion of an informative strategy. We then move to a formal discussion of secure strategies in Section~\ref{sec:secure_strategies}. In Section~\ref{sec: Simultaneously Informative and Secure Strategies}, we explore strategies that are simultaneously informative and either weakly or perfectly $\delta$-secure, discussing construction methods and examples in Section~\ref{subsec: construction methods and examples}. In Section~\ref{sec: Variant} we discuss a variant of the generalized Russian cards problem and present a solution using transversal designs. We discuss related work in Section~\ref{sec: Related Work}. Finally, we give some concluding remarks in Section~\ref{sec: Conclusion}.

\section{Combinatorial designs}
\label{sec:combinatorial_designs}

In this section, we  present fundamental definitions and standard results from the theory of combinatorial designs needed in this paper. For general references on this material, we refer the reader to Stinson~\cite{S03} and Colbourn and Dinitz~\cite{CD}. All results stated in this section without proof can be found in~\cite{S03,CD}. 

\subsection{$t$-designs:}
\begin{definition}\label{t-design}
Let $v$, $k$, $\lambda$, and $t$ be positive integers with $v > k \geq t$.
A \emph{$t$-$(v,k,\lambda)$-design} is a pair $(X, \B)$ such that the following are satisfied:
	\begin{enumerate}
		\item $X$ is a set of $v$ elements called \emph{points}, 
		\item $\B$ is a collection (i.e., a multiset) of nonempty proper $k$-subsets of $X$, called \emph{blocks}, and
		\item every subset of $t$ distinct points from $X$ occurs in precisely $\lambda$ blocks. 
	\end{enumerate}
\end{definition}

\begin{definition} The design formed by taking $\lambda$ copies of every $k$-subset of a $v$-set as blocks is a $t$-$\left (v,k,\lambda \binom{v-t}{k-t} \right)$-design, called a \emph{trivial $t$-design}.
\end{definition}

\begin{definition}
A $t$-$(v,k,\lambda)$-design $(X, \B)$ is \emph{simple} if every block in $\B$
occurs with multiplicity one.
\end{definition}

\begin{remark} In the context of the generalized Russian cards problem, we will consider simple designs only, although we allow for multisets in Definition~\ref{t-design} for completeness.
\end{remark}

The following two theorems are standard results for $t$-designs:

\begin{theorem} \label{block.containing.Y} Let $(X,\B)$ be a $t$-$(v,k,\lambda)$-design. Let $Y \subseteq X$ such that $\lvert Y \rvert = s \leq t$. Then there are precisely \[\lambda_s = \frac{\lambda\binom{v-s}{t-s}}{\binom{k-s}{t-s}}\] blocks in $\B$ that contain $Y$.
\end{theorem}

\begin{theorem}\label{allY noZ thm t designs}
 Let $(X,\B)$ be a $t$-$(v,k,\lambda)$-design. Let $Y\subseteq X$ and $Z \subseteq X$ such that $Y \cap Z = \emptyset$, $\lvert Y \rvert = i$, $\lvert Z \rvert = j$, and $i + j \leq t$. Then there are precisely \[\lambda_i^j = \frac{\lambda\binom{v-i-j}{k-i}}{\binom{v-t}{k-t}}\] blocks in $\B$ that contain all the points in $Y$ and none of the points in $Z$.
\end{theorem}

\begin{example}\label{3-design}
 A $3$-$(8,4,1)$-design. \[X = \{0,1,2,3,4,5,6,7\} \mbox{ and}\] 
\[\B =\{3456, 2567, 2347, 1457, 1367, 1246, 1235, 0467, 0357, 0245, 0236, 0156, 0134, 0127\}.\]
\end{example}

The blocks of the design in Example~\ref{3-design} are the planes of $\operatorname{AG}(3,2)$. This is an example of a special type of design known as a \emph{Steiner system}, which is a $t$-design with $\lambda =1$. 

\begin{definition} A $t$-$(v,k,1)$-design is called a \emph{Steiner system with parameters $t$, $k$, $v$} and is denoted by $\SkS(t,k,v)$.
\end{definition}

Steiner systems will be useful for constructing solutions to the generalized Russian cards problem. We list constructions and existence results from the literature which we will make use of here; interested readers may find further details in~\cite{CD}.

\begin{definition} A \emph{Steiner triple system of order $v$}, or \emph{$\STS(v)$}, is an $\SkS(2,3,v)$, i.e., a Steiner system in which $k=3$.
\end{definition}

\begin{theorem} It is known that an $\STS(v)$ exists if and only if $v \equiv 1,3 \bmod 6$, $v \geq 7$.	
\end{theorem}

\begin{definition} A \emph{Steiner quadruple system of order $v$} is an $\SkS(3,4,v)$. 
\end{definition}

\begin{theorem} Steiner quadruples exist if and only if $v \equiv 2, 4 \bmod 6$.
\end{theorem}

\begin{theorem} Known infinite families of $\SkS(t,k,v)$ are
	\begin{enumerate}
		\item $\SkS(2,q,q^n)$, for prime powers $q$, $n \geq 2$, called \emph{affine geometries};
		\item $\SkS(2,q+1,q^n+ \cdots + q+ 1)$, for prime powers $q$, $n \geq 2$, called \emph{projective geometries};
		\item $\SkS(2,q + 1,q^3+1)$, for prime powers $q$, called \emph{unitals};
		\item $\SkS(2, 2^r, 2^{r+s} +2^r -2^s)$, for $2 \leq r < s$, called \emph{Denniston designs};		
		\item $\SkS(3,q + 1,q^n+1)$, for prime powers $q$, $n \geq 2$, called \emph{spherical geometries} (or, when $n = 2$, \emph{inversive planes});
	\end{enumerate}
\end{theorem}
		
Only finitely many Steiner systems are known for $t  = 4, 5$ and none are known for $t > 5$. All known $\SkS(4, a, n)$ designs are \emph{derived designs} from $\SkS(5, a+1, n+1)$ designs, formed by choosing an element $x$, selecting all blocks containing $x$ and then deleting $x$ from these blocks. A list of the parameters for which these designs are known to exist may be found in Table~\ref{steiner:list} of Section~\ref{subsec: construction methods and examples}.

\begin{definition}\label{t-design large set}
A  \emph{large set of $t$-$(v,k,1)$-designs} is a set $\{(X,\B_1), \dots , (X,\B_N)\}$ of
$t$-$(v,k,1)$-designs (all of which have the same point set, $X$), 
in which every $k$-subset of $X$ occurs as a block in precisely
one of the $\B_i$s. That is, the  $\B_i$s form a partition of $\binom{X}{k}$.
\end{definition}

\begin{remark}
It is easy to prove that there must be exactly $N = \binom{v-t}{k-t}$ designs in a large set of $t$-$(v,k,1)$-designs.
\end{remark}
\begin{theorem}\label{large STS existence}
A large set of $\STS(v)$ exists if and only if
$v \equiv 1,3 \bmod 6$ and $v \geq 9$.
\end{theorem}

\begin{example} \label{large.set}
A large set of $\STS(9)$~\cite{MP97,Schreiber73}. \[X = \{1,2,3,4,5,6,7,8,9\} \mbox{ and } \B_1, \ldots, \B_7, \] where the $7$ block sets $\B_1, \ldots, \B_7$ are given by the rows of the following table:

\[
\begin{tabular}{cccccccccccc}
123 & 145 &  169 & 178 & 249 & 257 &  268 & 348  & 356 & 379 &  467 & 589  \\
124 &  136&  158 & 179 & 235 & 267 & 289&  349&  378&  457 & 468 & 569 \\
125 & 137 & 149&  168&  238 & 247&  269 & 346 &359 & 458 & 567 & 789 \\
126 & 139 & 148 & 157&  234 & 259  &278 & 358 & 367 & 456 &  479 &  689 \\
127 & 135 &146  &189 & 239  & 248  &256 & 347 & 368  &459 & 578 & 679 \\
128 & 134 & 159 & 167 & 236 & 245 & 279 & 357 &  389 & 469 & 478 & 568 \\
129  &138 & 147 & 156 & 237 & 246 & 258  &345 & 369 & 489 & 579 & 678
\end{tabular}
\]
\end{example}

The concept of \emph{balanced incomplete block designs (BIBDs)}, which are $t$-designs with $t=2$, will also be useful:

\begin{definition} A $2$-$(v,k,\lambda)$-design is also called a \emph{$(v,k,\lambda)$-balanced incomplete block design}, or \emph{$(v,k,\lambda)$-BIBD}.
\end{definition}

\begin{remark}
 In a $(v,k,\lambda)$-BIBD, every point occurs in precisely $r = {\lambda(v - 1)}/{(k-1)}$ blocks and the total number of blocks is $b = {vr}/{k}$.
\end{remark}

\begin{definition} Let $(X, \B)$ be a $(v, k, \lambda)$-BIBD. A \emph{parallel class} in $(X, \B)$ is a set of blocks that partition the point set. If $\B$ can be partitioned into parallel classes, we say $(X, \B)$ is a \emph{resolvable} BIBD.
\end{definition}

\emph{Symmetric designs} will also be useful in constructing solutions to the generalized Russian cards problem:

\begin{definition}
\label{def:symmetric intersection}
A \emph{symmetric BIBD} is a $(v,k,\lambda)$-BIBD in which there are $v$ blocks.
\end{definition}

\begin{theorem} \label{symmetric intersection}In a symmetric BIBD, any two blocks intersect in exactly $\lambda$ points.
\end{theorem}

Two infinite families of symmetric BIBDs, which we will use later, are 
\begin{enumerate}
	\item \emph{hyperplanes} in projective spaces, which are $\left ( \frac{q^{d+1}-1}{q-1}, \frac{q^d-1}{q-1},\frac{q^{d-1}-1}{q-1} \right )$-BIBDs, for prime powers $q$, and $d \geq 2$; and 
	\item \emph{Hadamard designs}, which are $\left (q, \frac{q-1}{2}, \frac{q-3}{4} \right)$-BIBDs, for odd prime powers $q$ satisfying $ q \equiv 3 \bmod 4$.
\end{enumerate}

\subsection{Transversal designs}
\label{subsec: transversal designs}

\begin{definition}Let $t$, $v$, $k$, and $\lambda$ be positive integers satisfying $k \geq t \geq 2$. A \emph{transversal design} $\TD_\lambda(t,k,v)$ is a triple $(X, \G, \B)$ such that the following properties are satisfied:
\begin{enumerate}
		\item $X$ is a set of $kv$ elements called \emph{points},
		\item $\G$ is a partition of $X$ into $k$ subsets of size $v$ called \emph{groups},
		\item $\B$ is a set of $k$-subsets of $X$ called \emph{blocks},
		\item any group and any block contain exactly one common point, and
		\item every subset of $t$ points from distinct groups occurs in precisely $\lambda$ blocks. 
	\end{enumerate}
\end{definition}

Many of the standard results for $t$-designs can be extended to transversal designs. The following terminology and results are useful:

\begin{definition} 
 Let $(X, \G, \B)$ be a $\TD_\lambda(t,k,v)$ and write $\G = \{G_j : 1 \leq j \leq k\}$. Suppose  $Z \subseteq X$ such that $\lvert Z| = i \leq k$ and $\lvert Z \cap G_j \rvert \leq 1$ for $1 \leq j \leq k$. We say $Z$ is a \emph{partial transversal of $\G$}. If $i = k$, then we say $Z$ is a \emph{transversal of $\G$}.
\end{definition}

\begin{definition}\label{def: group disjoint}
For a partial transversal $Z$ of $\G$, we let $G_Z = \{G_j \in \G : Z \cap G_j \neq \emptyset\}$ denote the set of groups that intersect $Z$. If  $Y$, $Z \subseteq X$ are partial transversals of $\G$ such that $G_Z \cap G_Y = \emptyset$, we say $Y$, $Z$ are \emph{group disjoint}. 
\end{definition}

\begin{theorem}\label{transversal params} Let $(X, \G, \B)$ be a $\TD_\lambda(t, k, v)$. Suppose $Y \subseteq X$ such that $\lvert Y \rvert =s \leq t$ and $Y$ is a  partial transversal of $\G$. Then there are exactly $\lambda_s = \lambda v^{t-s}$ blocks containing all the points in $Y$. 
\end{theorem}

\begin{proof} Fix a subset of $t-s$ groups disjoint from $Y$, say $G_1^\prime, \ldots, G_{t-s}^\prime$. Consider a $t$-subset $X$ consisting of all the points from $Y$ and one point from each of  $G_1^\prime, \ldots, G_{t-s}^\prime$. In particular, there are $v^{t-s}$ such $t$-subsets $X$, and each such $X$ occurs in precisely $\lambda$ blocks. Note that every block that contains $Y$ is a transversal of $\G$, so every such block contains exactly one such $t$-subset $X$.  Therefore $Y$ occurs in precisely $\lambda v^{t-s}$ blocks, as desired.
\end{proof}

\begin{theorem}\label{transversal security} Let $(X, \G, \B)$ be a $\TD_\lambda(t,k,v)$. Suppose $Y$, $Z \subseteq X$ are group disjoint partial transversals of $\G$ such that $\lvert Y \rvert =i$, $\lvert Z \rvert = j$, and $i + j \leq t$. Then there are exactly 
\[\lambda_i^j = \lambda v^{t-i-j}(v-1)^j \]
blocks in $\B$ that contain all the points in $Y$ and none of the points in $Z$.
\end{theorem}

\begin{proof}
Consider the set of groups $G_Z$ that intersect $Z$. There are $(v-1)^j$ subsets $X$ such that $X$ consists of all the points from $Y$ and one point from each group in $G_Z$, but $X$ contains no points from $Z$. Each such $(i+j)$-subset $X$ occurs in precisely $\lambda_{i+j}$ blocks by Theorem~\ref{transversal params}. Therefore there are $\lambda_{i+j}(v-1)^j = \lambda v^{t-i-j} (v-1)^j$ blocks that contain all the points of $Y$ but none of the points of $Z$.
\end{proof}

We can also apply the notion of \emph{large sets} to transversal designs:
\begin{definition}\label{def: large TD} A \emph{large set of $\TD_\lambda(t,k,v)$ on the point set $X$ and group partition $\G$} is a set $\{(X, \G, \B_1), \ldots, (X, \G, \B_N)\}$ of $\TD_\lambda(t,k,v)$ in which every set of $k$ points from distinct groups of $X$ occurs as a block in precisely one of the $\B_i$s. 
\end{definition}

\begin{remark}It is easy to see that there must be $N = \frac{v^k}{\lambda v^t}$ transversal designs in a large set of $\TD_\lambda(t,k,v)$.
\end{remark}

Transversal designs are equivalent to \emph{orthogonal arrays}:
\begin{definition}
Let $t, v, k$, and $\lambda$ be positive integers satisfying $k \geq t \geq 2$. An \emph{orthogonal array} $\OA_\lambda(t,k,v)$ is a pair $(X, D)$ such that the following properties are satisfied:
\begin{enumerate}
		\item $X$ is a set of $v$ elements called \emph{points},
		\item $D$ is a $\lambda v^t$ by $k$ array whose entries are elements of $X$, and
		\item within any $t$ columns of $D$, every $t$-tuple of points occurs in precisely $\lambda$ rows. 
	\end{enumerate}
\end{definition}

\begin{example}\label{ex: OA} An $\OA_1(2, 4,3)$.
\[
\begin{array}{cccc}
1 & 1 & 1 & 1 \\
1 & 2 & 3 & 3 \\
1 & 3 & 2 & 2 \\
2 & 1 & 2 & 3 \\
2 & 2 & 1 & 2 \\
2 & 3 & 3 & 1 \\
3 & 1 & 3 & 2 \\
3 & 2 & 2 & 1 \\
3 & 3 & 1 & 3
\end{array}
\]
\end{example}

It is easy to see the correspondence between orthogonal arrays and transversal designs. Suppose $(X, D)$ is an $\OA_\lambda(t,k,v)$. We define a bijection $\phi$ between the rows $r_j$ of $D$ and the blocks $B_j$ of a $\TD_\lambda(t,k,v)$ as follows. For each row $r_j = [x_{j1} x_{j2} \cdots x_{jk}]$ of $D$, let \[\phi(r_j)= \{(x_{j1},1),  (x_{j2}, 2), \ldots (x_{jk}, k)\}= B_j\] define a block $B_j$. Define $G_i = \{1,\ldots, v\} \times \{i\}$ for $1 \leq i \leq k$. Then $(X \times \{1,\ldots, k\}, \G, \B)$ is a $\TD_\lambda(t,k,v)$ with $\G = \{G_i : 1 \leq i \leq k\}$ and $\B = \{B_j : 1 \leq j \leq \lambda v^t\}$.

\begin{example} The blocks of the $\TD _1(2,4,3)$ obtained from the $\OA_1(2, 4,3)$ in Example~\ref{ex: OA}:
\[
\begin{array}{ccccc}
B_1 \colon (1,1) & (1,2) & (1,3) & (1,4) \\
B_2 \colon (1, 1) & (2,2) & (3,3) & (3,4) \\
B_3 \colon (1, 1) & (3,2) & (2,3) & (2,4) \\
B_4 \colon (2, 1) & (1,2) & (2,3) & (3,4) \\
B_5 \colon (2, 1) & (2,2) & (1,3) & (2,4) \\
B_6 \colon (2, 1) & (3,2) & (3,3) & (1,4) \\
B_7 \colon (3, 1) & (1,2) & (3,3) & (2,4) \\
B_8 \colon (3, 1) & (2,2) & (2,3) & (1,4) \\
B_9 \colon (3, 1) & (3,2) & (1,3) & (3,4)
\end{array}
\]
\end{example}

The above construction method can be reversed for an arbitrary $\TD_\lambda(t,k,v)$, say $(X, \G, \B)$. To see this, note that we can relabel the points such that $X = \{1, \ldots, v\} \times \{1, \ldots, k\}$ and $\G = \{G_i : 1 \leq i \leq k\}$. Then the fact that any block and any group must contain exactly one common point implies that for each $B \in \B$, we can form the $k$-tuple $(b_1, \ldots, b_k)$, where $b_i \in B \cap G_i$ for $1 \leq i \leq k$. We can form an orthogonal array $\OA_\lambda(t,k,v)$ by taking all of these $k$-tuples as rows.

\begin{definition}A \emph{large set of $\OA_\lambda(t,k,v)$ on the point set $X$} is a set of
$\OA_\lambda(t,k,v)$, say $\{(X,D_1), \dots , (X,D_N)\}$,  in which every $k$-tuple of elements from $X$ occurs as a row in precisely
one of the $D_i$s. That is, the  $D_i$s form a partition of the set $X^k$ of $k$-tuples with entries from $X$. 
\end{definition}

\begin{remark}It is easy to see that there must be $N = \frac{v^k}{\lambda v^t}$ orthogonal arrays in a large set of $\OA_\lambda(t,k,v)$.
\end{remark}

A useful type of orthogonal array is a \emph{linear array}, especially for constructing large sets: 

\begin{definition}
Let $(X,D)$ be an $\OA_\lambda(t,k,v)$. We say $(X, D)$  is \emph{linear} if $X = \mathbb{F}_q$ for some prime power $q$ and the rows of $D$ form a subspace of $(\mathbb{F}_q)^k$ of dimension $\log_q \lvert D \rvert$.
\end{definition}

Linear orthogonal arrays (and hence the corresponding transversal designs) are easy to construct. In particular, the following is a useful construction method.

\begin{theorem}\label{linear OA thm} Suppose $q$ is a prime power and $k$ and $\ell$ are positive integers. Suppose $M$ is an $\ell$ by $k$ matrix over $\mathbb{F}_q$ such that every set of $t$ columns of $M$ is linearly independent. Then $(X, D)$ is a linear $\OA_{q^{\ell-t}}(t,k,q)$, where $D$ is the $q^\ell$ by $k$ matrix formed by taking all linear combinations of the rows of $M$.
\end{theorem}

Let $q$ be a prime power and for every $x \in \mathbb{F}_q$, let $\vec{x} = [1,x,x^2,\ldots,x^{t-1}] \in (\mathbb{F}_q)^t$ for some integer $t \geq 2$. Construct the $t$ by $q$ matrix $M$ by taking the columns to be the vectors $(\vec{x})^T$ for every $x \in \mathbb{F}_q$, where here $(\vec{x})^T$ means the transpose of $\vec{x}$. Applying Theorem~\ref{linear OA thm} to $M$ yields the following result:

\begin{corollary}\label{linear OA p.p.thm}Let $t \geq 2$ be an integer and let $q$ be a prime power. Then there exists a linear $\OA_1(t,q,q)$.
\end{corollary}

The following result is immediate.

\begin{corollary}\label{linear TD p.p. thm} Let $t \geq 2$ be an integer and let $q$ be a prime power. Then there exists a linear $\TD_1(t,q,q)$.
\end{corollary}

\begin{remark} The constructions discussed in Corollaries~\ref{linear OA p.p.thm}~and~\ref{linear TD p.p. thm} are known as \emph{Reed-Solomon codes}~\cite{S03}.
\end{remark}

We now discuss how to construct a large set of orthogonal arrays from a ``starting'' linear orthogonal array. Suppose $(X,D)$ is a linear $\OA_\lambda(t,k,v)$. We can obtain a large set of orthogonal arrays (and therefore transversal designs) from $(X, D)$ by taking the set of cosets of $D$ in $(\mathbb{F}_q)^k$. In particular, $D$ is a subspace of $(\mathbb{F}_q)^k$, so the cosets of $D$ form a partition of $(\mathbb{F}_q)^k$.

\section{Terminology and notation}
\label{subsec: Preliminary Notation and Examples}

We review the terminology and notation established by Swanson and Stinson~\cite{cardsSwansonS12}. Throughout, we let $\binom{X}{t}$ denote the set of 
$\binom{n}{t}$ $t$-subsets of $X$, where $t$ is a positive integer.

Let $X$ be a deck of $n$ cards. In an \emph{$(a,b,c)$-deal} of $X$, Alice is dealt a \emph{hand} $H_A$ of $a$ cards, Bob is dealt a hand $H_B$ of $b$ cards, and Cathy is dealt a hand $H_C$ of $c$ cards, such that $a + b + c = n$. That is, it must be the case that $H_A \cup H_B \cup H_C = X$. We assume  these hands are random and dealt by some external entity. 

An \emph{announcement} by Alice is a subset of $\binom{X}{a}$ containing Alice's current hand, $H_A$. More generally, Alice chooses a set of $m$ announcements $\A_1, \A_2, \dots  , \A_m \subseteq \binom{X}{a}$ satisfying $\bigcup_{i=1}^m \A_i = \binom{X}{a}$. For every $H_A \in \binom{X}{a}$, we define $g(H_A)= \{ i : H_A \in \A_i\}$, i.e., the set of possible announcements for Alice given the hand $H_A$. Alice's \emph{announcement strategy}, or simply \emph{strategy}, consists of a probability distribution $p_{H_A}$ defined on $g(H_A)$, for every $H_A \in \binom{X}{a}$. 

In keeping with Kerckhoffs' principle, we assume the set of announcements and probability distributions are fixed ahead of time and public knowledge.  For a given hand $H_A \in \binom{X}{a}$, Alice randomly chooses an index $i \in g(H_A)$ according to the 
probability distribution $p_{H_A}$. Alice broadcasts the integer $i$ to specify her announcement $\A_i$. Without loss of generality, we assume that $p_{H_A}(i) > 0$ for all $i \in g(H_A)$.

For the purposes of this paper, we assume there exists some constant $\gamma$ such that $\left \lvert g(H_A)\right \rvert   = \gamma$  for every $H_A$ and that every probability distribution $p_{H_A}$ is \emph{uniform}; such strategies are termed \emph{$\gamma$-equitable}, or simply \emph{equitable}. Throughout, we use the phrase \emph{$(a,b,c)$-strategy} $\D$ to denote a strategy for an $(a,b,c)$-deal, where $\D$ is the associated set of possible announcements for Alice. 

The following notation is useful in discussing the properties of a given strategy $\D$.
For any subset $Y \subseteq X$ and any announcement $\A \in \D$, we define \[\mathcal{P} \left (Y, \A \right ) = \left \{H_A \in \A : H_A \cap Y = \emptyset \right\}.\] That is, $\mathcal{P} \left (Y, \A \right )$ is the set of hands of $\A$ that do not intersect the subset $Y$. 

\section{Informative strategies}
\label{sec:informative_strategies}

Suppose we have an $(a,b,c)$-deal and Alice chooses announcement $\A$ from the set $\D$ of possible announcements. From Bob's point of view, the set of possible hands for Alice given Alice's announcement $\A$ and Bob's hand $H_B \in \binom{X}{b}$ is
\[ \mathcal{P} \left (H_B, \A \right) = \left \{ H_A \in \A : H_A \cap H_B = \emptyset \right \} .\]
We say Alice's strategy is \emph{informative for Bob} provided that 
\begin{equation}
\label{inf-Bob}
\left \lvert \mathcal{P} \left (H_B,\A \right) \right \rvert   \leq 1 
\end{equation}
for all $H_B \in \binom{X}{b}$ and for all $\A \in \D$. That is, if Equation~(\ref{inf-Bob}) is satisfied, Bob can determine the set of $a$ cards that Alice holds from Alice's announcement. In particular, this implies that Bob can announce Cathy's hand, thereby informing Alice of the card deal as well. Specified on the level of individual announcements, we say an announcement $\A$ is \emph{informative} provided $\left \lvert \mathcal{P}\left (H_B,\A \right) \right \rvert   \leq 1$ for any hand $H_B \in \binom{X}{b}$.

The following theorem, first shown by Albert et al.~\cite{AAADH05}, is a useful equivalence condition for informative announcements:
\begin{theorem}{\cite{AAADH05}}
\label{infor.thm}
The announcement $\A$ is informative for Bob if and only if there do not
exist two distinct sets $H_A$, $H_A^\prime \in \A$ such that 
$\left \lvert H_A \cap  H_A^\prime \right \rvert   \geq a-c$.
\end{theorem}

The following is an immediate corollary.
\begin{corollary}
\label{NC-cor}
Suppose 
there exists a strategy for Alice that is informative
for Bob. Then $a > c$.
\end{corollary}

We make the following observation, which follows directly from Theorem~\ref{infor.thm} and the definition of a $t$-design.

\begin{corollary}\label{info.corollary} Let $n = a + b + c$. Suppose $a > c$ and each announcement $\A$ in an $(a,b,c)$-strategy is a $t$-$(n,a,1)$-design for some $t$, where $t \leq a-c$. Then the strategy is informative for Bob.
\end{corollary}

It is possible to have informative $(a,b,c)$-strategies using announcements which are $t$-designs with $\lambda > 1$. In particular, Theorem~\ref{infor.thm} indicates that the block intersection properties of the chosen design are relevant to whether or not the strategy is informative. If every announcement is a symmetric BIBD, for example, then the strategy is guaranteed to be informative when $a-c > \lambda$. This is because the intersection of any two blocks in a symmetric BIBD contains exactly $\lambda$ points, as stated in Theorem~\ref{symmetric intersection}.

We make one more observation relating combinatorial designs and informative strategies.
\begin{lemma}\label{new.info.lemma} Let $n = a + b + c$. Suppose $a > c$ and each announcement $\A$ in an $(a,b,c)$-strategy $\D$ is a $t$-$(n,a,\lambda)$-design for some $t$ and $\lambda$, where $t \geq a-c$. If $\D$ is informative for Bob, then $t=a-c$ and $\lambda =1$ for all $\A \in \D$.
\end{lemma}

\begin{proof}Consider an announcement $\A \in \D$. If $\lambda > 1$, then there exist two blocks whose intersection has cardinality at least $t \geq a-c$. This contradicts Theorem~\ref{infor.thm}, so $\lambda =1$, as desired.

If $t> a-c$, then from Theorem~\ref{block.containing.Y}, there are 
\[ \frac{v-(t-1)}{k-(t-1)} > 1
\]
blocks that contain $t -1$ fixed points. Since $t -1 \geq a-c$, this contradicts Theorem~\ref{infor.thm}, so $t = a-c$, as desired.
\end{proof}

\section{Secure strategies}
\label{sec:secure_strategies}

We provide the general security definitions and state the equivalent combinatorial characterization of secure equitable strategies from Swanson and Stinson~\cite{cardsSwansonS12}.

\begin{definition}\label{security.gen.def}
Let $1 \leq \delta \leq a$.
\begin{enumerate}
\item Alice's strategy is \emph{weakly $\delta$-secure against Cathy} provided that for any announcement $\A$, for any $H_C \in \binom{X}{c}$ such that $\mathcal{P}\left(H_C,\A \right) \neq \emptyset$, and for any $\delta^\prime$-subset $Y \subseteq X \backslash H_C$ where $1 \leq \delta^\prime \leq \delta$, it holds that 
\[0 < \Pr \left [ Y \subseteq H_A \mid \A,H_C \right ] < 1.\]
Weak security means that, from Cathy's point of view, any set of $\delta$ or fewer elements from $X \backslash H_C$
may or may not be held by Alice.

\item \label{perfect delta sec. def}Alice's strategy is \emph{perfectly $\delta$-secure against Cathy} provided that for any announcement $\A$, 
for any $H_C \in \binom{X}{c}$ such that $\mathcal{P}\left(H_C,\A \right) \neq \emptyset$, 
and for any $\delta^\prime$-subset $Y \subseteq X \backslash H_C$ where $1 \leq \delta^\prime \leq \delta$, it holds that 
\[ \Pr \left [ Y \subseteq H_A \mid \A,H_C \right] = \frac{\binom{a}{\delta^\prime}}{\binom{a+b}{\delta^\prime}}.\]
Perfect security means that, from Cathy's point of view, the probability that 
any set of $\delta$ or fewer cards from $X \backslash H_C$
is held by Alice is a constant.
\end{enumerate}
\end{definition}

\begin{remark} The requirement that $\mathcal{P}\left(H_C,\A \right) \neq \emptyset$ ensures that it is feasible (within the constraints of the announcement) for Cathy to hold the given hand $H_C$; we sometimes refer to a hand that satisfies this condition as a \emph{possible hand for Cathy}.
\end{remark}

Swanson and Stinson~\cite{cardsSwansonS12} show that in an equitable strategy any hand 
$H_A \in \mathcal{P}\left (H_C,\A \right)$ is equally likely from Cathy's point of view:

\begin{lemma}{\cite{cardsSwansonS12}}
\label{equitable.lem}
Suppose that Alice's strategy is $\gamma$-equitable, Alice's announcement is $\A$,
$H_C \in \binom{X}{c}$ and $H_A \in \mathcal{P}\left (H_C,\A \right)$.  Then
\begin{align}
\label{equitable.eq}
\Pr \left [H_A \mid H_C,\A \right] = \frac{1 }{ \left \lvert \mathcal{P}\left (H_C,\A \right) \right \rvert}.
\end{align}
\end{lemma}

Swanson and Stinson~\cite{cardsSwansonS12} also establish the following equivalent combinatorial conditions:

\begin{theorem}\label{equitable1} {\cite{cardsSwansonS12}}
Suppose that Alice's strategy is $\gamma$-equitable. Then the following hold:
\begin{enumerate}
\item \label{weak1} Alice's strategy is weakly $\delta$-secure against Cathy if and only if,  for any announcement $\A$, for any $H_C \in \binom{X}{c}$ such that $\mathcal{P}\left(H_C,\A \right) \neq \emptyset$, and for any $\delta^\prime$-subset $Y \subseteq X \backslash H_C$ where $1 \leq \delta^\prime \leq \delta$, it holds that 
\[ 1 \leq \left \lvert \left \{ H_A \in \mathcal{P}\left (H_C,\A \right) : Y \subseteq H_A \right \} \right \rvert   \leq \left \lvert \mathcal{P}\left (H_C,\A \right)\right \rvert   - 1.\]

\item \label{strong equitable1}
Alice's strategy is perfectly $\delta$-secure against Cathy if and only if, for any announcement $\A$ and
for any $H_C \in \binom{X}{c}$ such that $\mathcal{P}\left (H_C,\A \right) \neq \emptyset$, 
it holds that
\[ \left \lvert \{ H_A \in \mathcal{P} \left (H_C,\A \right) : Y \subseteq H_A \}\right \rvert   = 
\frac{\binom{a}{\delta} \, \left \lvert \mathcal{P}\left (H_C,\A \right)\right \rvert }{ \binom{a+b}{\delta}}\]
for any $\delta$-subset $Y \subseteq X \backslash H_C$.
\end{enumerate}
\end{theorem}

We have the following elementary result:
\begin{lemma}\label{single element lemma} Consider an $(a,b,c)$-strategy $\D$ that is weakly 1-secure. Then for all $\A \in \D$ and $x \in X$, we have $\mathcal{P}\left(\{x\}, \A \right) \neq \emptyset$.
\end{lemma}

\begin{proof}
We proceed by contradiction. Suppose $\mathcal{P}\left (\{x\}, \A \right) = \emptyset$ for some $\A \in \D$ and $x \in X$. Then $x$ occurs in every hand of $\A$. That is, if Alice announces $\A$, then Alice must hold $x$. In particular, this implies that Cathy's hand, say $H_C$, does not contain $x$ and $\Pr \left [x \in H_A \mid \A, H_C \right ] = 1$.
\end{proof}

Here is a sufficient condition for an equitable strategy to be perfectly $1$-secure 
against Cathy, first shown by Swanson and Stinson~\cite{cardsSwansonS12}:

\begin{lemma}{\cite{cardsSwansonS12}}
\label{2-design}
Let $n =a +b + 1$. Suppose that each announcement $\A$ in an equitable $(a,b,1)$-strategy $\D$ is a $2$-$(n,a,\lambda)$-design for some $\lambda$. Then the strategy is perfectly $1$-secure against Cathy.
\end{lemma}

In fact, the condition that every announcement $\A$ be a $2$-$(n,a,\lambda)$-design for some $\lambda$ is also a necessary condition for an equitable $(a,b,1)$-strategy to be perfectly $1$-secure, as the following Theorem shows.

\begin{theorem}\label{new 2-design thm} Let $n = a + b + 1$. Suppose we have an equitable $(a,b,1)$-strategy $\D$ that is perfectly 1-secure against Cathy. Then every announcement $\A \in \D$ is a $2$-$(n,a,\lambda)$-design for some $\lambda$.
\end{theorem}

\begin{proof} First observe that since Cathy holds only one card, Lemma~\ref{single element lemma} immediately implies that any element $x \in X$ is a possible hand for Cathy. Consider an announcement $\A \in \D$. We proceed by showing that every pair of distinct elements $x$, $y \in X$ occurs in a constant number of hands of $\A$.

Let $x \in X$. Define $r_x$ to be the number of hands of $\A$ containing $x$. We proceed by counting $r_x$ in two different ways. On the one hand, we immediately have 
\begin{align}
r_x = \left \lvert \A \right \rvert   - \left \lvert \mathcal{P}\left (\{x\},\A \right)\right \rvert . \label{four}
\end{align}

On the other hand, we can relate $r_x$ to $\mathcal{P}\left (\{y\}, \A \right)$ for any $y \neq x \in X$ as follows. By Theorem~\ref{equitable1}.\ref{strong equitable1}, $x$ occurs $\frac{a}{a+b}\left \lvert \mathcal{P}\left(\{y\},\A\right)\right \rvert $ times in $\mathcal{P}\left (\{y\},\A \right)$. In particular, this is the number of times $x$ occurs in a hand of $\A$ without $y$. That is, letting $\lambda_{xy}$ denote the number of times $x$ occurs together with $y$ in a hand of $\A$, we have
\begin{align}
r_x = \lambda_{xy} + \frac{a}{a+b} \left \lvert \mathcal{P}\left(\{y\}, \A \right) \right \rvert  .\label{three}
\end{align}
This gives us
\begin{align}
\left \lvert \A \right \rvert   = \lambda_{xy} + \frac{a}{a+b} \left \lvert \mathcal{P}\left(\{y\},\A \right) \right \rvert   +  \left \lvert \mathcal{P}\left (\{x\}, \A \right)\right \rvert . \label{one}
\end{align}

Now, following the same logic for $y$, we also have
\begin{align}
\left \lvert \A \right \rvert   = \lambda_{xy} + \frac{a}{a+b} \left \lvert \mathcal{P}\left (\{x\},\A \right) \right \rvert  +  \left \lvert \mathcal{P}\left (\{y\},\A \right)\right \rvert . \label{two}
\end{align}

Equating Equations~(\ref{one})~and~(\ref{two}) shows that $\left \lvert \mathcal{P}\left (\{x\},\A \right)\right \rvert $ is independent of the choice of $x \in X$. That is, $r_x$ is independent of $x$ (by Equation~(\ref{four})), so every point of $X$ occurs in a constant number of hands of $\A$, say $r$ hands. Moreover, Equation~(\ref{three}) then gives
\[
\lambda_{xy} = r - \frac{a}{a+b} \left \lvert \mathcal{P}\left(\{y\},\A \right) \right \rvert  
= r - \frac{a}{a+b} \left( \left \lvert A \right \rvert  - r \right),
\]
so $\lambda_{xy}$ is independent of $x$ and $y$. That is, every pair of points $x$, $y \in X$ occurs a constant number of times, which we denote by $\lambda$. This implies $\A$ is  a  $2$-$(n,a,\lambda)$-design. 
\end{proof}

The relationship between combinatorial designs and strategies that satisfy our notion of perfect security is quite deep. We now generalize the results from Swanson and Stinson~\cite{cardsSwansonS12} and Theorem~\ref{new 2-design thm} above to account for perfect $\delta$-security and card deals with $c \geq 1$. We begin with a generalization of Lemma~\ref{2-design} that shows that in an equitable $(a,b,c)$-strategy, if each announcement is a $t$-design with block size $a$, the strategy satisfies perfect $(t-c)$-security.

\begin{theorem}\label{t-design security}
Let $n = a + b + c$. Suppose that each announcement $\A$ in an equitable $(a,b,c)$-strategy $\D$ is a $t$-$(n,a,\lambda)$-design for some $t$ and $\lambda$, where $c \leq t-1$. Then the strategy is perfectly $(t-c)$-secure against Cathy.
\end{theorem}

\begin{proof}
Consider an announcement $\A \in \D$ and a possible hand $H_C$ for Cathy.  Since $c \leq t$, Theorem~\ref{allY noZ thm t designs} implies there are 
\[
\left \lvert \mathcal{P}\left (H_C, \A \right)\right \rvert   = \frac{\lambda \binom{n-c}{a}}{\binom{n-t}{a-t}} = \frac{\lambda \binom{a+b}{a}}{\binom{n-t}{a-t}}
\]
blocks in 
$\A$ that do not contain any of the points of $H_C$.

Let $\delta \leq t-c$. Then Theorem~\ref{allY noZ thm t designs} also implies that each set of $\delta$ points $x_1,\dotsc,x_\delta \in X \backslash H_C$ is contained in precisely 
\[ \left \lvert \{ H_A \in \mathcal{P}\left(H_C, \A \right) : x_1,\dotsc,x_{\delta} \in H_A \}\right \rvert   = \frac{\lambda \binom{n-\delta-c}{a-\delta}}{\binom{n-t}{a-t}}= \frac{\lambda \binom{a+b-\delta}{a-\delta}}{\binom{n-t}{a-t}}\] of these blocks.

Thus, for any set of $\delta$ points $x_1,\dotsc,x_\delta \in X \backslash H_C$, we have  
\[
\frac{\left \lvert \mathcal{P}\left(H_C,\A \right)\right \rvert  }{\left \lvert \{ H_A \in \mathcal{P}\left(H_C, \A \right) : x_1,\dotsc,x_{\delta} \in H_A \}\right \rvert  } = \frac{(a+b)!(a-\delta)!}{a!(a+b - \delta)!}
= \frac{\binom{a+b}{\delta}}{\binom{a}{\delta}},
\]
so Condition~\ref{strong equitable1}~of~Theorem~\ref{equitable1} is satisfied.
\end{proof}

We approach a true generalization of Theorem~\ref{new 2-design thm} incrementally for readability. For deals satisfying $c=1$, we have the following necessary condition for an equitable strategy to be perfectly $\delta$-secure.

\begin{theorem}\label{new design thm} Let $n = a + b + 1$. Suppose we have an equitable $(a,b,1)$-strategy $\D$ that is perfectly $\delta$-secure against Cathy. Then every announcement $\A \in \D$ is a $(\delta+1)$-$(n,a,\lambda)$-design for some $\lambda$.\end{theorem}

\begin{proof}
We proceed by induction on $\delta$. The base case ($\delta = 1$) is shown in Theorem~\ref{new 2-design thm}. 

Consider an announcement $\A \in \D$. For a subset $Y \subseteq X$, let $\lambda_{Y}$ denote the number of hands of $\A$ that contain $Y$. We show $\A$ must be a $(\delta+1)$-design as follows. 

Suppose we have $Y \subseteq X$, where $\left \lvert Y \right \rvert   = \delta+1$. Pick an element $y \in Y$. Since $c=1$, we have $\mathcal{P}\left(\{y\}, \A \right) \neq \emptyset$ by Lemma~\ref{single element lemma}, so $\{y\}$ is a possible hand for Cathy. Since $\D$ is equitable and perfectly $\delta$-secure, we have (by Theorem~\ref{equitable1})
\[
\left \lvert \{ H_A \in \mathcal{P}\left (\{y\},\A \right) : Y \backslash \{y\} \subseteq H_A \} \right \rvert   = \frac{\binom{a}{\delta} \, \left \lvert\mathcal{P}\left(\{y\},\A \right)\right \rvert }{ \binom{a+b}{\delta}}.
\]
Moreover, since perfect $\delta$-security implies perfect $1$-security, $\left \lvert \mathcal{P}\left(\{y\},\A \right) \right \rvert  $ is independent of $y$, as shown in the proof of Theorem~\ref{new 2-design thm}. That is, the number of hands of $\A$ that contain the $\delta$-subset $Y \backslash \{y\}$ but do not contain $y$ is independent of the choice of $Y$ and $y \in Y$, i.e., is some constant, say $s$.

Now, $\D$ must be perfectly $(\delta-1)$-secure (since  $\D$ is perfectly $\delta$-secure), so by the inductive hypothesis, $\A$ is a $\delta$-$(n,a,\lambda^\prime)$-design for some $\lambda^\prime$. Therefore, the number of hands of $\A$ that contain the $\delta$-subset $Y \backslash \{y\}$ is precisely $\lambda^\prime$.

We have
\begin{align*}
& \lambda_{Y \backslash \{y\}} = \lambda_{Y} + \frac{\binom{a}{\delta} \, \left \lvert \mathcal{P}\left(\{y\},\A \right)\right \rvert }{ \binom{a+b}{\delta}}\\
\iff & \lambda^\prime = \lambda_Y + s.
\end{align*}
Therefore, $\lambda_Y$ is some constant independent of $Y$, so every $(\delta+1)$-subset occurs in a constant number of hands of $\A$, say $\lambda$. This implies $\A$ is a $(\delta+1)$-$(n,a,\lambda)$-design, as desired.
\end{proof}

We are now ready to give a combinatorial characterization of general $(a,b,c)$-strategies that are equitable and perfectly $\delta$-secure for some $\delta \geq 1$. We give an inductive proof that relies on Theorem~\ref{new design thm} as the base case.

\begin{theorem}\label{new design thm.gen} Let $n = a + b + c$. Suppose we have an equitable $(a,b,c)$-strategy $\D$ that is perfectly $\delta$-secure against Cathy. Then every announcement $\A \in  \D$ is a $(c+ \delta)$-$(n,a,\lambda)$-design for some $\lambda$.\end{theorem}

\begin{proof} 
We proceed by induction on $c$. The base case $c=1$ is shown in Theorem~\ref{new design thm}. 

Let $y \in X$ and define $X^\prime = X \backslash \{y\}$. For any $\A \in \D$, we define $\A^\prime$ to be the set of hands in $\A$ that do not contain $y$; that is, $\A^\prime = \mathcal{P} \left(\{y\},\A \right)$.

We then define an $(a, b, c-1)$-strategy $\D^\prime$ by
\[\D^\prime = \left \{\A^\prime :  \A \in \D \right\}.
\]

We now show $\D^\prime$ is perfectly $\delta$-secure. Suppose Cathy holds a $(c-1)$-subset $Y \subseteq X^\prime$ satisfying $\mathcal{P} \left(Y,\A^\prime \right) \neq \emptyset$ for some $\A^\prime = \mathcal{P} \left(\{y\},\A \right) \in \D^\prime$. In particular, note that if no such $\A^\prime$ exists, then $\D^\prime$ is trivially perfectly $\delta$-secure.

Consider a $\delta$-subset $Z \subseteq X^\prime \backslash Y = X \backslash (Y \cup \{y\})$. We wish to count the number of hands in $\mathcal{P} \left(Y, \A^\prime \right )$ that contain $Z$. Now, $\mathcal{P} \left(Y, \A^\prime \right) = \mathcal{P} \left(Y \cup \{y\},\A \right)$, so  $\mathcal{P} \left(Y \cup \{y\},\A \right) \neq \emptyset$ and hence $Y \cup \{y\}$ is a possible hand for Cathy in the original strategy $\D$. Since $\D$ is perfectly $\delta$-secure, we see that (by Theorem~\ref{equitable1})
\[
 \left \lvert \{ H_A \in \mathcal{P} \left(Y \cup \{y\},\A \right) : Z \subseteq H_A \} \right \rvert   = 
\frac{\binom{a}{\delta} \, \left \lvert \mathcal{P} \left(Y \cup \{y\},\A \right)\right \rvert }{ \binom{a+b}{\delta}},
\]
which together with the fact that $\mathcal{P} \left(Y,\A^\prime \right) = \mathcal{P} \left(Y \cup \{y\},\A \right)$, immediately implies $\D^\prime$ is perfectly $\delta$-secure. Moreover, since $\D^\prime$ is a perfectly $\delta$-secure $(a,b,c-1)$-strategy, we have by the inductive hypothesis that every announcement  $\A^\prime \in \D^\prime$ is a $(c-1 + \delta)$-$(n-1,a,\lambda_y)$-design for some $\lambda_y$, where $\lambda_y$ may depend on $y$.

That is, every $(c-1 + \delta)$-subset of $X \backslash \{y\}$ occurs in $\lambda_y$ hands of $\A^\prime = \mathcal{P} \left(\{y\},\A \right)$. We show this implies  $\D$ is a $(c-1+\delta)$-perfectly secure $(a, b+c-1, 1)$-strategy by counting the total number (with repetition) of $(c-1+\delta)$-subsets of $\mathcal{P} \left(\{y\},\A \right)$ in two ways and then applying Theorem~\ref{equitable1}. First we observe that there are $\binom{a+b+c-1}{c-1+\delta}$ ways of picking a $(c-1+\delta)$-subset of $X \backslash \{y\}$, and each of these subsets occurs in $\lambda_y$ hands of  $\mathcal{P} \left(\{y\},\A \right)$. Second, we observe there are $\lvert \mathcal{P} \left(\{y\},\A \right) \rvert$ possible hands for Alice (from Cathy's perspective), and each of these possible hands yields $\binom{a}{c-1+\delta}$ many $(c-1+\delta)$-subsets. 

This gives, for any $(c-1+\delta)$-subset $Z^\prime \subseteq X \backslash \{y\}$, 
\[
 \left \lvert \{ H_A \in \mathcal{P} \left( \{y\}, \A \right) : Z^\prime \subseteq H_A \} \right \rvert  = \lambda_y = \frac{\binom{a}{c-1+\delta} \lvert \mathcal{P} \left(\{y\},\A \right) \rvert}{\binom{a+b+c-1}{c-1+\delta}}.
\]

Since we chose $y$ to be an arbitrary element of $X$, by applying Theorem~\ref{equitable1} we see that $\D$ is a $(c-1+\delta)$-perfectly secure $(a, b+c-1, 1)$-strategy. Then the base case (Theorem~\ref{new design thm}) implies that every announcement $\A \in \D$ is a $(c+\delta)$-$(n,a,\lambda)$-design for some $\lambda$, as desired.
\end{proof}

Theorem~\ref{new design thm.gen} immediately implies the following bound on the security parameter $\delta$ for equitable strategies:
\begin{corollary} 
Suppose we have an equitable $(a,b,c)$-strategy $\D$ that is perfectly $\delta$-secure against Cathy. Then $\delta \leq a-c$.
\end{corollary}

\begin{remark} If we have an equitable $(a,b,c)$-strategy $\D$ that is perfectly $\delta$-secure against Cathy, where $\delta = a-c$, then each announcement $\A \in \D$ is an $a$-design. In fact, since every $a$-subset of $X$ must appear a constant number of times in each $\A$, we see that each $\A$ is a \emph{trivial $a$-design}. In this case, we see Alice's strategy is not informative for Bob.
\end{remark}

Together, Theorem~\ref{t-design security} and Theorem~\ref{new design thm.gen} show a direct correspondence between $t$-designs and equitable announcement strategies that are perfectly $\delta$-secure for some $\delta$ satisfying $\delta \leq t-c$. We state this result in the following theorem for clarity.

\begin{theorem}\label{main} A $\gamma$-equitable $(a,b,c)$-strategy $\D$ on card deck $X$ that is perfectly $\delta$-secure against Cathy is equivalent to a set of $(c+\delta)$-designs with point set $X$ and block size $a$ having the property that every $a$-subset of $X$ occurs in precisely $\gamma$ of these designs.
\end{theorem}

\section{Simultaneously informative and secure strategies}
\label{sec: Simultaneously Informative and Secure Strategies}

In general, we want to find an $(a,b,c)$-strategy (for Alice) 
that is simultaneously informative for
Bob and (perfectly or weakly) $\delta$-secure against Cathy. We first consider informative strategies that provide security for individual cards and then consider informative strategies that provide security for multiple cards.

The following was first shown by Albert et al.~\cite{AAADH05}:

\begin{theorem}{\cite{AAADH05}}
If $a \leq c+1$, then there does not exist a strategy for Alice that is simultaneously informative for Bob and weakly $1$-secure against Cathy.
\end{theorem}

It is worth observing that a strategy that is not informative for Cathy implies,  for any announcement $\A$ by Alice and possible hand $H_C \in \binom{X}{c}$ such that $\mathcal{P}\left(H_C, \A \right) \neq \emptyset$ , that $\left \lvert \mathcal{P}\left(H_C, \A \right)\right \rvert \geq 2$. That is, there must exist distinct $H_A, H_A^\prime \in \mathcal{P}\left(H_C,\A \right)$.  Following the same technique as in the proof of Lemma~\ref{infor.thm}, this implies $\left \lvert H_A \cap H_A^\prime \right \rvert  \geq a-b$. If in addition the strategy is informative for Bob, by Lemma~\ref{infor.thm} we have $a - c >\left \lvert H_A \cap H_A^\prime \right \rvert  \geq a -b$, so $c < b$. This gives us the following result (which is also discussed by Albert et al.~\cite{AAADH05}):

\begin{theorem}If $c \geq b$, then there does not exist a strategy for Alice that is simultaneously informative for Bob and weakly $1$-secure against Cathy.
\end{theorem}

We now focus on $(3,n-4,1)$-deals and examine the relationship between informative and perfectly 1-secure strategies and Steiner triple systems. 

The following is an immediate consequence of Theorem~\ref{new 2-design thm} and Lemma~\ref{new.info.lemma}.
\begin{corollary}\label{Steiner}Suppose $(a,b,c) = (3,n-4,1)$ and suppose that Alice's strategy is equitable, informative for Bob, and perfectly 1-secure against Cathy. Then every announcement is a Steiner triple system.
\end{corollary}

In fact, any $(a, b, a-2)$-strategy that is informative, equitable, and perfectly 1-secure also satisfies $c=1$ (and hence $a=3$). This result was first shown in Swanson and Stinson~\cite{cardsSwansonS12}, but the proof provided here is greatly simplified.

\begin{theorem}
\label{main.thm}
Consider an $(a,b,c)$-deal such that $a-c=2$. Suppose that Alice's strategy is equitable, informative for Bob, and perfectly 1-secure against Cathy. Then $a=3$ and $c=1$.
\end{theorem}

\begin{proof}Theorem~\ref{new design thm.gen} implies that every announcement is an $(a-1)$-design. Since $c \geq 1$, we have $a-1 \geq a-c$, so we may apply Lemma~\ref{new.info.lemma}. This implies $a-1 = a-c$, so we have $c=1$, as desired.
\end{proof}

Our proof technique works for the generalizations of Theorem~\ref{main.thm} and Corollary~\ref{Steiner} shown in Swanson and Stinson~\cite{cardsSwansonS12} as well. That is, strategies that are equitable, informative for Bob, and perfectly $(a-c-1)$-secure against Cathy  must satisfy $c=1$ and each announcement must be an $(a-1)$-$(n,a,1)$-design, also known as a Steiner system $\SkS(a-1,a,n)$.

\begin{theorem}
\label{main.thm1}
Consider an $(a,b,c)$-deal. Suppose that Alice's strategy is equitable, informative for Bob, and perfectly $(a-c-1)$-secure against Cathy. Then $c=1$.
\end{theorem}

\begin{proof} The proof is identical to the proof of Theorem~\ref{main.thm}.
\end{proof}

\begin{corollary}
\label{Steiner.gen} Let $n = a + b + 1$. Consider an equitable $(a,b,1)$-strategy that is informative for Bob and perfectly $(a-2)$-secure against Cathy. Then every announcement is a Steiner system $\SkS(a-1,a,n)$.
\end{corollary}

\begin{proof} The fact that every announcement is an $(a-1)$-design follows immediately from Theorem~\ref{new design thm.gen}. To see that $\lambda =1$, we may apply  Lemma~\ref{new.info.lemma}. This is easy to see, however: since every $(a-1)$-subset occurs $\lambda$ times, the fact that the strategy is informative for Bob implies $\lambda =1$.
\end{proof}

In fact, we can use Theorem~\ref{new design thm.gen} and Lemma~\ref{new.info.lemma} to derive the following bound on the security parameter $\delta$ for perfectly $\delta$-secure and informative strategies, which helps put the above results in context.

\begin{corollary}  Suppose we have an equitable $(a,b,c)$-strategy that is perfectly $\delta$-secure against Cathy and informative for Bob. Then $\delta \leq a-2c$.
\end{corollary}

\begin{proof}
If the strategy is perfectly $\delta$-secure, then by Theorem~\ref{new design thm.gen}, every announcement is a $(c+\delta)$-design. Now, if $c + \delta < a-c$ holds, then $\delta < a-2c$, as desired. If $c + \delta \geq a-c$, then since the strategy is informative for Bob, we can apply Lemma~\ref{new.info.lemma}. This yields $c+ \delta = a-c$, so we have $\delta = a-2c$ in this case.
\end{proof}

\section{Construction methods and examples}
\label{subsec: construction methods and examples}
Theorem~\ref{t-design security} indicates that we can use $t$-designs to construct equitable strategies that are perfectly $\delta$-secure against Cathy for $\delta = t-c$, where $c \leq t-1$. In fact, so long as we use $t$-designs with $\lambda = 1$ and $c \leq a-t$, such a strategy will also be informative for Bob (Corollary~\ref{info.corollary}). This is a very interesting result, as we can use a single ``starting design'' to obtain equitable strategies that are informative for Bob and perfectly $\delta$-secure against Cathy. We give a general method for this next. First we require some definitions.  

\begin{definition}
Suppose that $\mathcal{D} = (X,\B)$ is a $t$-$(v,k,\lambda)$-design.
An \emph{automorphism} of $\mathcal{D}$ is a permutation $\pi$ of $X$ such that
$\pi$ fixes the multiset $\B$. We denote the collection of all automorphisms of $\mathcal{D}$
by $\mathsf{Aut}(\mathcal{D})$.
\end{definition}

\begin{remark}
It is easy to see that $\mathsf{Aut}(\mathcal{D})$
is a subgroup of the symmetric group $S_{\left \lvert X\right \rvert }$.
\end{remark}

\begin{theorem}\label{aut thm gen} Suppose $\mathcal{D} = (X, \mathcal{B})$ is a $t$-$(n,a,1)$-design. Then there exists a $\gamma$-equitable $(a,n-a-c,c)$-strategy with $m$ announcements that is informative for Bob and perfectly $(t-c)$-secure against Cathy for any choice of $c$ such that $c \leq \min \{t-1,a-t\}$, where $m = n! / \! \left \lvert \mathsf{Aut}(\mathcal{D}) \right \rvert $ and $\gamma = \left. m \middle / \!  \binom{n-t}{a-t} \right.$.
\end{theorem}

\begin{proof}
Let the symmetric group $S_n$ act on $\mathcal{D}$. We obtain a set of designs isomorphic to $\mathcal{D}$, which are the announcements in our strategy. Since each announcement is a $t$-$(n,a,1)$-design, the resulting scheme is perfectly $(t-c)$-secure against Cathy by Theorem~\ref{t-design security}. Furthermore, since $a-c \geq t$ and $\lambda =1$, no two blocks have more than $a-c -1$ points in common, so Theorem~\ref{infor.thm} implies the scheme is informative for Bob.

The total number of designs $m$ is equal to $n! / \! \left \lvert \mathsf{Aut}(\mathcal{D})\right \rvert $ (as this is the index of $\mathsf{Aut}(\mathcal{D})$ in $S_{n}$). To see that $\gamma = \left. m \middle / \!  \binom{n-t}{a-t} \right.$, consider a fixed $t$-subset $A$ of $X$. Then in particular, there are $\binom{n-t}{a-t}$ possible blocks of size $a$ that contain $A$. Now, every one of the $m$ designs contains exactly one of these $\binom{n-t}{a-t}$ blocks, and these $\binom{n-t}{a-t}$ blocks occur equally often among the $m$ designs. Thus, a given block $B$ occurs in $\left. m \middle / \!  \binom{n-t}{a-t} \right.$ of the designs, as desired.
\end{proof}

\begin{remark} Theorem~\ref{aut thm gen} is a generalization of a result in Swanson and Stinson~\cite{cardsSwansonS12}, in which the case $c = 1$ is treated.
\end{remark}

\begin{remark}\label{starting design technique} The technique described in Theorem~\ref{aut thm gen} shows how to use a single ``starting design'' $\mathcal{D}$ on $n$ points to construct a strategy that inherits its properties from $\mathcal{D}$. That is, the strategy obtained by letting the symmetric group $S_n$ act on $\mathcal{D}$ will be informative and perfectly $\delta$-secure if $\mathcal{D}$ is an informative announcement that satisfies Condition~\ref{perfect delta sec. def} of Definition~\ref{security.gen.def} for the fixed announcement $\mathcal{D}$.
\end{remark}

We now discuss some other constructions of strategies using results from design theory, including some applications of Remark~\ref{starting design technique}. All constructions discussed may be found in Colbourn and Dinitz~\cite{CD}.

It is clear that we can use any Steiner triple system, or $2$-$(n,3,1)$-design, as a starting design to obtain an equitable $(3, n-4, 1)$-strategy that is informative for Bob and perfectly 1-secure against Cathy. It is known that an $\STS(n)$ exists if and only if $n \equiv 1,3 \bmod 6$, $n \geq 7$. We state this result in the following Corollary.

\begin{corollary} There exists an equitable $(3, n-4, 1)$-strategy for Alice that is informative for Bob and perfectly $1$-secure against Cathy for any integer $n$ such that $n \equiv 1,3 \bmod 6$, $n \geq 7$. 
\end{corollary}

Similarly, \emph{Steiner quadruple systems}, or $3$-$(n,4,1)$-designs, exist if an only if $n \equiv 2, 4 \bmod 6$, which yields the following result:

\begin{corollary} There exists an equitable $(4, n-5, 1)$-strategy for Alice that is informative for Bob and perfectly $1$-secure against Cathy for any integer $n$ such that $n \equiv 2,4 \bmod 6$. 
\end{corollary}

More generally, we can use any Steiner system $\SkS(t,a,n)$ as a starting design to obtain an equitable $(a, n-a-c, c)$-strategy that is perfectly $(t-c)$-secure against Cathy for $c \leq \min\{t-1,a-t\}$. Known infinite families of $\SkS(2,a,n)$ include \emph{affine geometries}, \emph{projective geometries}, \emph{unitals}, and \emph{Denniston designs}~\cite{CD}, which together give the following result:

\begin{corollary} Let $q$ be a prime power and $\ell \geq 2$. There exist the following equitable strategies that are perfectly $1$-secure against Cathy:
	\begin{enumerate}
		\item  A $(q, q^\ell -q-1,1)$-strategy (constructed from affine geometries);
		\item A $(q+1, q^\ell + \cdots + q^2  - 1, 1)$-strategy (constructed from projective geometries);
		\item A $(q+1, q^3-q-1, 1)$-strategy (constructed from unitals); and
		\item A $(2^r, 2^{r+s} - 2^s -1, 1)$-strategy, for $2 \leq r < s$ (constructed from Denniston designs).
	\end{enumerate}
\end{corollary}

In fact, we can use the same method to construct equitable $(a,b,c)$-strategies that are perfectly $\delta$-secure against Cathy, informative for Bob, and \emph{allow Cathy to hold more than one card}. Such a solution to the generalized Russian cards problem has not been proven to exist in the literature. We next give an infinite class of equitable and perfectly 1-secure strategies where Cathy holds two cards. 

\begin{example}Consider the  \emph{inversive plane} with $q=2^3$; this is a $3$-$(65,9,1)$-design. The construction method in Theorem~\ref{aut thm gen} yields an equitable $(9,55,1)$-strategy that is perfectly 2-secure against Cathy and informative for Bob and (more interestingly) a $(9,54,2)$-strategy that is perfectly 1-secure against Cathy and informative for Bob.
\end{example}

It is known that $3$-$(q^2+1, q+1, 1)$-designs (or inversive planes) exist whenever $q$ is a prime power. This gives us the following result.

\begin{corollary} There exists an equitable $\left (q+1, q^2-q-2, 2 \right)$-strategy that is informative for Bob and perfectly $1$-secure against Cathy and an equitable $\left (q+1, q^2-q-1, 1 \right)$-strategy that is informative for Bob and perfectly $2$-secure against Cathy, for every prime power $q \geq 4$.
\end{corollary}

More generally, we can use \emph{spherical geometries}, which are $3$-$(q^n +1, q+1, 1)$-designs (or, equivalently, $\SkS(3, q+1, q^n+1)$) for $q$ a prime power and $n \geq 2$ to construct strategies allowing Cathy to hold two cards:

\begin{corollary} There exists an equitable $\left (q+1, q^n-q-2, 2 \right)$-strategy that is informative for Bob and perfectly $1$-secure against Cathy and an equitable $\left (q+1, q^n-q-1, 1 \right)$-strategy that is informative for Bob and perfectly $2$-secure against Cathy, for every prime power $q$ and $n \geq 2$.
\end{corollary}

However, only finitely many Steiner $t$-designs are known for $t > 3$ and none are known for $t > 5$. Table~\ref{steiner:list} lists strategies resulting from known Steiner $5$- and $4$-designs; see~\cite{CD} for examples of these designs. All known $\SkS(4, a, n)$ designs are \emph{derived designs} from $\SkS(5, a+1, n+1)$ designs, formed by choosing an element $x$, selecting all blocks containing $x$ and then deleting $x$ from these blocks.

\begin{example} As Table~\ref{steiner:list} indicates, a $\SkS(5,8,24)$ exists. This design and its derived $\SkS(4,7,23)$ are called the \emph{Witt designs}. In particular, the $\SkS(5,8,24)$ implies that for an $(8, 13, 3)$-deal, it is possible to achieve perfect $2$-security. This is the only construction of which the authors are aware that achieves perfect security for $c > 2$. 
\end{example}

\begin{table}
\caption{Perfectly $(t-c)$-secure strategies from Steiner $t$-designs for $t =4,5$}
\label{steiner:list}
\centering
\begin{tabular}{| r || r || c || r || r || c |} \hline
$5$-design		& $(a,b,c)$-strategy	& $5-c$	& Derived $4$-design	& $(a,b,c)$-strategy	& $4-c$\\ \hline \hline
$\SkS(5,8,24)$		& $(8, 15, 1)$		& $4$		& $\SkS(4,7,23)$		& $(7, 15, 1)$		& $3$\\
				& $(8,14, 2)$		& $3$ 		& 				& $(7, 14, 2)$		& $2$\\
				& $(8, 13, 3)$		& $2$		& 				& $(7, 13, 3)$		& $1$\\ \hline
$\SkS(5,7,28)$		& $(7, 20, 1)$		& $4$		& $\SkS(4,6,27)$		& $(6, 20, 1)$		& $3$\\
				& $(7, 19, 2)$		& $3$		& 				& $(6, 19, 2)$		& $2$\\ \hline
$\SkS(5, 6,12)$ 		& $(6, 5, 1)$		& $4$		& $\SkS(4, 5,11)$ 		& $(5, 5, 1)$		& $3$\\\hline
$\SkS(5, 6, 24)$		& $(6, 17, 1)$		& $4$		& $\SkS(4, 5, 23)$		& $(5, 17, 1)$		& $3$\\\hline
$\SkS(5, 6, 36)$		& $(6, 29, 1)$		& $4$		& $\SkS(4, 5, 35)$		& $(5, 29, 1)$		& $3$ \\\hline
$\SkS(5, 6, 48)$		& $(6, 41, 1)$		& $4$		& $\SkS(4, 5, 47)$		& $(5, 41, 1)$		& $3$\\\hline
$\SkS(5, 6, 72)$		& $(6, 65, 1)$		& $4$		& $\SkS(4, 5, 71)$		& $(5, 65, 1)$		& $3$\\\hline
$\SkS(5, 6, 84)$		& $(6, 77, 1)$		& $4$		& $\SkS(4, 5, 83)$		& $(5, 77, 1)$		& $3$ \\\hline
$\SkS(5, 6, 108)$		& $(6, 101, 1)$		& $4$		& $\SkS(4, 5, 107)$	& $(5, 101, 1)$		& $3$ \\\hline
$\SkS(5, 6, 132)$		& $(6, 125, 1)$		& $4$		& $\SkS(4, 5, 131)$	& $(5, 125, 1)$		& $3$ \\\hline
$\SkS(5, 6, 168)$		& $(6, 161, 1)$		& $4$		& $\SkS(4, 5, 167)$	& $(5, 161, 1)$		& $3$\\\hline
$\SkS(5, 6,  244)$		& $(6, 137, 1)$		& $4$		& $\SkS(4, 5,  243)$	& $(5, 137, 1)$		& $3$ \\\hline
\end{tabular}
\end{table}

We next discuss existence results for \emph{optimal} strategies. As shown in Swanson and Stinson~\cite{cardsSwansonS12}, the number of announcements $m$ in an informative $(a,b,c)$-strategy must satisfy $m \geq \binom{n-a+c}{c}$. A strategy is \emph{optimal} if $m = \binom{n-a+c}{c}$. The following result by Swanson and Stinson~\cite{cardsSwansonS12} follows immediately from the existence of large sets of Steiner triples, discussed in Remark~\ref{large STS existence}, and Lemma~\ref{2-design}.

\begin{theorem}\label{large steiner strategy} {\cite{cardsSwansonS12}}
Suppose $(a,b,c) = (3,n-4,1)$, where $n \equiv 1,3 \bmod 6$, $n > 7$. Then there exists an optimal strategy for Alice that is informative for Bob and perfectly $1$-secure against Cathy.
\end{theorem}

\begin{example} Consider the large set of $\STS(9)$ from Example~\ref{large.set}. This set of announcements is an optimal $(3,5,1)$ strategy that is perfectly 1-secure against Cathy and informative for Bob.
\end{example}

As discussed before Theorem~\ref{large steiner strategy}, if we can construct a large set of $2$-$(n,3,1)$-designs, this set forms an optimal strategy that is informative and perfectly $1$-secure, and a large set of $\STS(n)$ exists whenever $n \equiv 1,3 \bmod 6$ and $n > 7$.  However, there are certain choices of $n$ for which there is a particularly nice construction for a large set of $\STS(n)$, such that it would be easy for Alice and Bob to create this large set on their own. We forego the details of this construction, which is due to Schreiber~\cite{Schreiber73}, but remark that this construction method applies whenever each prime divisor $p$ of $n-2$ has the property that the order of $(-2)$ modulo $p$ is congruent to $2$ modulo $4$.

Two other types of designs that can be used to construct informative and perfectly 1-secure strategies where Cathy holds one card are \emph{hyperplanes} in projective spaces and \emph{Hadamard designs}. For a discussion of these constructions, we refer the reader to Stinson~\cite{S03}. We have the following results.

\begin{corollary} There exists an equitable $\left (\frac{q^d-1}{q-1}, q^d-1 , 1\right)$-strategy that is informative for Bob and perfectly $1$-secure against Cathy, where $q \geq 2$ is a prime power and $d \geq 2$ is an integer.
\end{corollary}

\begin{proof}It is known that there exists a symmetric  $\left ( \frac{q^{d+1}-1}{q-1}, \frac{q^d-1}{q-1},\frac{q^{d-1}-1}{q-1} \right )$-BIBD $\mathcal{D}$ for every prime power $q $ and integer $d \geq 2$. The design $\mathcal{D}$ is a hyperplane in a projective space (or, in the case $d=2$, a finite projective plane). Let the symmetric group $S_n$  act on $\mathcal{D}$ as in the proof of Theorem~\ref{aut thm gen}, where $n= (q^{d+1}-1)/(q-1)$, to obtain Alice's strategy.

Lemma~\ref{2-design} immediately implies that this strategy is perfectly 1-secure against Cathy. To see that this strategy is informative, recall that the intersection of two blocks in a symmetric BIBD has size $\lambda = (q^{d-1}-1)/(q-1)$. It is easy to see that the strategy will be informative provided $a - c > \lambda$, which is the case here.
\end{proof}

\begin{corollary}There exists an equitable $\left (\frac{q-1}{2}, \frac{q-1}{2} , 1 \right)$-strategy that is informative for Bob and perfectly $1$-secure against Cathy, where $q \equiv 3 \bmod 4$ is an odd prime power.
\end{corollary}

\begin{proof}
It is known that there exists a symmetric $\left (q, \frac{q-1}{2}, \frac{q-3}{4} \right)$-BIBD $\mathcal{D}$ for every odd prime power $q$ such that $q \equiv 3 \bmod 4$. The design $\mathcal{D}$ is a Hadamard design.  Let the symmetric group $S_q$ act on $\mathcal{D}$ as in the proof of Theorem~\ref{aut thm gen} to obtain Alice's strategy.

Lemma~\ref{2-design} immediately implies that this strategy is perfectly 1-secure against Cathy. To see that this strategy is informative, recall that the intersection of two blocks in a symmetric BIBD has size $\lambda = (q-3)/4$. It is easy to see that the strategy will be informative provided $a - c > \lambda$, which is the case here.
\end{proof}

\begin{remark} Any symmetric BIBD may be used to construct equitable strategies that are perfectly $1$-secure against Cathy for $c=1$. If $\mathcal{D}$ is a symmetric $2$-$(n,a,\lambda)$-design, the \emph{order} of $\mathcal{D}$ is $a - \lambda$. The block intersection property we need to guarantee that the strategy is informative is that the order is greater than 1, which will always be the case. Colbourn and Dinitz~\cite{CD} list known families of symmetric BIBDs.
\end{remark}

\subsection{Cord\'{o}n-Franco et al.\ geometric protocol}
\label{subsec: geom}
Cord\'{o}n-Franco et al.~\cite{CDFS13} present a protocol based on hyperplanes that yields informative and weakly $\delta$-secure equitable $(a,b,c)$-strategies for arbitrary $c, \delta > 0$ and appropriate parameters $a$ and $b$. The geometric protocol is stated as follows.

\begin{protocol}[Geometric Protocol~\cite{CDFS13}]
Let $p$ be a prime power and let $d$ and $s < p $ be positive integers. Let $X$ be a deck of $p^{d+1}$ cards  and suppose we have an $(a,b,c)$-deal such that $a = sp^d$. Given a hand $H_A \in \binom{X}{a}$, the set of possible announcements for Alice is the set of bijections from $X$ to $\AG_{d+1}(p)$ satisfying the condition that $H_A$ maps to the union of $s$ parallel hyperplanes in $\AG_{d+1}(p)$. For every $H_A \in \binom{X}{a}$, assume Alice picks uniformly at random from the set of possible bijections. 
\end{protocol}

In particular, the geometric protocol defines an equitable strategy in which Cathy may hold more than one card. We analyze when the geometric protocol achieves perfect, rather than weak, security, whereas Cord\'{o}n-Franco et al.~\cite{CDFS13} show that the general case achieves weak $s$-security for a card deck of size $p^{d+1}$, where $a=sp^d$, if $c < sp^d-s^2p^{d-1}$ and $\max \{c+s,cs\} \leq p$.

We now translate the geometric protocol into our model.

\begin{observation} Let $\D$ be the strategy defined by the geometric protocol. An announcement $\A_i \in \D$ is equivalent to the set of all possible unions of $s$ parallel hyperplanes.
\end{observation}

We first consider general results from design theory with respect to an announcement in the above strategy $\D$. Let us view $X$ as the set of points in $\AG_{d+1}(p)$, and let $\B$ be denote the set of all hyperplanes in $\AG_{d+1}(p)$. It is well known that $(X, \B)$ is a resolvable $\left (p^{d+1}, p^d,  \lambda \right)$-BIBD, where $\lambda = {(p^d-1)}/{(p-1)}$. Moreover, each point has degree $r =({p^{d+1}-1})/({p-1})$, and there are $r$ equivalance classes of parallel hyperplanes, each of size $p$. Let $\Pi_1, \ldots, \Pi_r$ denote these equivalence classes.  
For each $i$, where $1 \leq i \leq r$, let the blocks in $\Pi_i$ be denoted $B_i^j$, for $1 \leq j \leq p$. 

We define the design $(X, \C)$ by forming a collection of all possible unions of $s$ parallel hyperplanes. Stated formally, let $\mathcal{D}$ be the set of all $s$-subsets of a set $Y$, where $\lvert Y \rvert = p$. For each $i$, where $1 \leq i \leq r$, and for each $D \in \mathcal{D}$, define
\[C_{i,D} = \bigcup_{j \in D} B_i^j.\]
We then let 
\[ \C = \{C_{i,D} : 1 \leq i \leq r, D \in \mathcal{D}\}.\]

As discussed in Stinson et al.~\cite{StinsonST13}, this construction $(X,\C)$ is a $\left (p^{d+1}, sp^d, \lambda^\prime \right )$-BIBD, where $\lambda^\prime = \binom{p-1}{s-1} \frac{sp^d-1}{p-1}$.  The above immediately implies the following observation:

\begin{observation}\label{thm: geom.2-design}Let $p$ be a prime power and let $d \geq 1$ be a positive integer. Let $X$ be a deck of $p^{d+1}$ cards and fix an $(a,b,c)$-deal with $a =sp^d$. Then in the strategy $\D$ defined by the geometric protocol, each announcement $\A$ is a $2$-$\left (p^{d+1}, sp^d, \lambda\right )$-design, where $\lambda = \binom{p-1}{s-1} \frac{sp^d-1}{p-1}$. In particular, there are $\binom{p}{s} ({p^{d+1}-1})/({p-1})$ possible hands for Alice in each $\A$.
\end{observation}

Observation~\ref{thm: geom.2-design} and Theorem~\ref{t-design security} imply that the Geometric Protocol achieves perfect $1$-security when Cathy holds one card, i.e., for $(sp^d,p^{d+1}-sp^d-1,1)$-deals where $p$ is a prime power and $s < p$.

Moreover, as shown by Stinson et al.~\cite{StinsonST13}, the design $(X, \C)$ is a $3$-design precisely when $p = 2s$, so $p$ must be an even prime power. In this case, $(X, \C)$ is a $3$-$\left (p^{d+1}, p^{d+1}/2, \lambda^{\prime\prime} \right )$-design, where 
\[\lambda^{\prime\prime} =  \binom{p-1}{p/2-1}\frac{p^{d+1} - 4}{4(p-1)}.
\]

That is, for card decks and deals satisfying certain parameters, the strategy defined by the geometric protocol is a $3$-design. This implies that we can sometimes achieve perfect $2$-security for deals in which Cathy holds one card, or perfect $1$-security for deals in which Cathy holds two cards. We state the result in the following theorem for clarity.

\begin{theorem} 
Let $p$ be a prime power and let $d \geq 1$ be a positive integer. Let $X$ be a deck of $p^{d+1}$ cards and fix an $(a,b,c)$-deal with $a =sp^d$. Then the geometric protocol gives perfect $1$-security with $c=2$ (and therefore also perfect $2$-security with $c=1$) if and only if $p = 2^\ell$ for some positive integer $\ell$ and $s = 2^{\ell-1}$.
\end{theorem}

\section{The transversal Russian cards problem} \label{sec: Variant}

In this section, we consider a variation of the generalized Russian cards problem, which we name the \emph{transversal Russian cards problem}, in which we change the manner in which the cards are dealt. Our motivation for restricting the deal is to widen the solution space. Since the generalized Russian cards problem requires a suitable set of $t$-designs to maximize security against Cathy---and constructing $t$-designs for $t > 2$ is in general quite difficult---we explore certain types of deals where suitable constructions are more readily available. An added advantage of our deal restriction is that in this new framework, we can view Alice's hand as an $a$-tuple over an alphabet of size $v$. If Alice's hand represents a secret key, this variation is more in keeping with traditional key agreement schemes in cryptography, as typically secret keys are tuples rather than sets.

Suppose our deck $X$ consists of $n =va$ cards, where $v$ and $a$ are positive integers such that $v > a$. Rather than allowing Alice, Bob, and Cathy to have any hand of the appropriate size, we first split the deck $X$ into $a$ piles, each of size $v$. Alice is given a hand $H_A$ of $a$ cards, such that she holds exactly one card from each pile. Cathy's hand $H_C$ of $c$ cards is assumed to contain no more than one card from each pile. The remainder of the deck becomes Bob's hand, $H_B$. We will refer to this type of deal as a \emph{transversal $(a,b,c)$-deal} (or simply, a \emph{transversal deal}). Observe that we can use the same framework for this problem as for the original; we have only placed a limitation on the set of possible hands Alice, Bob, and Cathy might hold. The necessary modifications to the security definitions and the definition of an informative strategy are straightforward.

This variant admits a nice solution using \emph{transversal designs}; we refer the reader to Section~\ref{subsec: transversal designs} for the relevant definitions and a discussion of these designs. In the context of a transversal design $\TD_\lambda(t,a,v)$, we can view the piles of cards as the groups $G_1, \dotsc, G_a$ of the design. In this case, Alice's hand is a transversal and Cathy's hand is a partial transversal of $G_1, \dotsc, G_a$. Note that Cathy therefore only considers transversals as possible hands for Alice. When we discuss weak (or perfect) $\delta$-security, we are interested in the probability (from Cathy's point of view) that Alice holds partial transversals of order $\delta$.

We first show Theorem~\ref{infor.thm} holds for this variant of the Russian cards problem:
\begin{theorem}
\label{infor.thm.trans}
The announcement $\A$ is informative for Bob if and only if there do not
exist two distinct sets $H_A, H_A^\prime \in \A$ such that 
$\left \lvert H_A \cap  H_A^\prime \right \rvert  \geq a-c$.
\end{theorem}

\begin{proof}
Suppose there exist two distinct sets $H_A, H_A^\prime \in \A$ such that 
$\left \lvert H_A \cap  H_A^\prime \right \rvert  \geq a-c$. We proceed by constructing a card deal consistent with the announcement $\A$ such that $ \{ H_A, H_A^\prime\} \subseteq  \mathcal{P}\left(H_B, \A \right),$ which implies the announcement is not informative for Bob.  

Write $\left \lvert H_A \cap  H_A^\prime \right \rvert   = \ell$. Let Alice's hand be $H_A$, so it is possible for Alice to announce $\A$. Let Cathy's hand contain all the cards in $H_A^\prime$ that are not also contained in $H_A$; this is possible since $c \geq a- \ell$. Then Bob's hand $H_B$ contains all the remaining cards. In particular, we have $H_B \cap (H_A \cup  H_A^\prime) = \emptyset$, so  $ \{ H_A, H_A^\prime\} \subseteq  \mathcal{P}\left(H_B,\A \right)$, as desired.

Conversely, suppose $\{ H_A, H_A^\prime\} \subseteq  \mathcal{P}\left(H_B,\A \right)$, where $H_A \neq H_A^\prime$.
Then $\left \lvert H_A \cup  H_A^\prime \right \rvert  \leq n-b = a+c$, and hence $\left \lvert H_A \cap  H_A^\prime \right \rvert  \geq a-c$.
\end{proof}

In light of Theorem~\ref{infor.thm.trans}, the following result is straightforward.

\begin{theorem}Consider a transversal $(a,b,c)$-deal and suppose that each announcement in an equitable $(a,b,c)$-strategy is a $\TD_1(t,a,v)$ satisfying $t \leq a-c$. Then the strategy is informative for Bob.
\end{theorem}

We can use an argument similar to that of Swanson and Stinson~\cite{cardsSwansonS12} to derive a lower bound on the size of Alice's announcement.

\begin{theorem}\label{bound.thm2}
Consider a transversal $(a,b,c)$-deal. Suppose $a >c$ and there exists a strategy for Alice that is informative for Bob. Then the number of announcements $m$ satisfies $m \geq v^c$.
\end{theorem}

\begin{proof}
Fix a set of cards $X^\prime$ of size $a-c$, no two of which are from the same pile. There are $v^c$ possible hands for Alice that contain $X^\prime$. These hands must occur in different announcements, by Theorem \ref{infor.thm.trans}. Therefore $m \geq v^c$.
\end{proof}

As before, we refer to a strategy that meets this bound as \emph{optimal}. We have the following result.
\begin{theorem}
\label{LS.thm variant}
Consider a transversal $(a,b,c)$-deal and suppose that $a > c$. An optimal $(a,b,c)$-strategy for Alice that is informative for Bob is equivalent to a 
large set of $\TD_1(t,a,v)$, where $t = a-c$. 
\end{theorem}

\begin{proof}
Suppose there exists a large set of $\TD_1(a-c,a,v)$. Recall from Definition~\ref{def: large TD} that the set of  all blocks sets (i.e., possible announcements) in this large set form a partition of the set of all transversals and that there are precisely $v^c$ designs in such a set. Then it is easy to see that this immediately yields an optimal $(a,b,c)$-strategy for Alice that is informative for Bob.

Conversely, suppose there is an optimal $(a,b,c)$-strategy for Alice that is informative for Bob.
We need to show that every announcement is a $\TD_1(a-c,a,v)$. As in the proof of Theorem \ref{bound.thm2}, fix a set of cards $X^\prime$ of size $a-c$, no two of which are from the same pile. The $v^c$ possible hands for Alice that contain $X^\prime$ must occur in different
announcements. However, there are a total of $v^c$ announcements, so every announcement must
contain exactly one block that contains $X^\prime$.
\end{proof}

The following result shows how transversal designs with arbitrary $t$ can be used to achieve weak $\delta$-security for permissible parameters $\delta \leq t-c$. As in Definition~\ref{def: group disjoint}, for a transversal design $\TD_\lambda(t,a,v)$, say  $(X, \G, \B)$, and a partial transversal $Y$ of $\G$, we let $G_Y$ denote the set of groups of the transversal design that have nonempty intersection with the partial transversal $Y$.

\begin{theorem}\label{transversal design security} Consider a transversal $(a,b,c)$-deal and suppose that each announcement in an equitable $(a,b,c)$-strategy is a $\TD_\lambda(t,a,v)$ for some $t$ and $\lambda$, where $c \leq t-1$. Then the strategy is weakly $(t-c)$-secure against Cathy.
\end{theorem}

\begin{proof} Fix an announcement $\A$ for Alice. Suppose $\A$ is a $\TD_\lambda(t,a,v)$, say  $(X, \G, \B)$. Consider a possible hand $H_C$ for Cathy.  In particular, $H_C$ is a partial transversal of the groups $G_1, \dotsc, G_a \in \G$. 

Since $c \leq t$, Theorem~\ref{transversal security} implies there are 
\[
\left \lvert \mathcal{P}\left(H_C,\A \right)\right \rvert  = \lambda v^{t-c}(v-1)^c 
\]
blocks in 
$\A$ that do not contain any of the points of $H_C$.

Consider a partial transversal $Y$ of order $\delta \leq t-c$. Since $Y$ is not necessarily group disjoint from $H_C$, we must consider the number of groups which intersect both $Y$ and $H_C$. In particular, the $\delta$-subset $Y$ never occurs with any other cards from $G_Y \cap G_{H_C}$, by definition of transversal designs.

Let $\ell = \left \lvert G_{H_C} \middle \backslash G_{Y} \right \rvert $. That is, $\ell$ is the number of groups that do not intersect $Y$, but from which Cathy has cards.  Write $z_1,\dotsc,z_\ell$ for Cathy's cards from these $\ell$ groups. We wish to compute the number of blocks which contain all the points in $Y$ but miss all of the points of $H_C$. This is the same as the number of blocks that contain all the points in $Y$ but miss all the points in $\{z_1,\dotsc,z_\ell\}$. Since $\ell + \delta \leq t$, by Theorem~\ref{transversal security}, we have $\lambda v^{t-\ell-\delta}(v-1)^\ell$ such blocks.

That is, a given set of points $x_1,\dotsc,x_\delta \in X \backslash H_C$ that might be held by Alice is contained in precisely 
\[ \left \lvert \{ H_A \in \mathcal{P}\left(H_C, \A \right) : x_1,\dotsc,x_{\delta} \in H_A \}\right \rvert  = \lambda v^{t-\ell-\delta}(v-1)^\ell\] of the blocks in $\mathcal{P}\left(H_C, \A \right)$, where $\ell = \left \lvert G_{H_C} \middle \backslash G_{\{x_1,\dotsc,x_\delta\}} \right \rvert .$

Thus, for any partial transversal of $\delta$ distinct points $x_1,\dotsc,x_\delta \in X \backslash H_C$, we have  
\[
\frac{\left \lvert \{ H_A \in \mathcal{P}\left(H_C, \A \right) : x_1,\dotsc,x_{\delta} \in H_A \}\right \rvert }{\left \lvert \mathcal{P}\left(H_C, \A \right)\right \rvert  } = \frac{\lambda v^{t-\ell-\delta}(v-1)^\ell}{\lambda v^{t-c}(v-1)^c }
= \frac{1}{v^{\delta+\ell-c}(v-1)^{c-\ell}},
\]
so Condition~\ref{weak1}~of~Theorem~\ref{equitable1} is satisfied.
\end{proof}

\begin{remark}  We do not achieve perfect $(t-c)$-security in Theorem~\ref{transversal design security} because the number of hands of $\mathcal{P}\left(H_C, \A \right)$ containing a given partial transversal $Y$ of $\delta$ distinct points, where $\delta \leq t-c$, depends on $\ell = \left \lvert G_{H_C} \middle \backslash G_{Y} \right \rvert$. In fact, we cannot expect to achieve better security than that of the construction given in Theorem~\ref{transversal design security} for this variant of the generalized Russian cards problem. This is because the rules for the deal imply that for each pile from which Cathy holds a card, Cathy knows that Alice holds one of the other $(v-1)$ cards, and for every other pile, Cathy knows only that Alice holds one of the other $v$ cards.
\end{remark}

As discussed in Section~\ref{subsec: transversal designs}, large sets of transversal designs $\TD_\lambda(t,k,v)$ are easy to construct when you have a linear  $\TD_\lambda(t,k,v)$ ``starting design''. As stated in Theorem~\ref{linear OA p.p.thm}, a linear $\TD_1(t,q,q)$ exists whenever the point set $X = (\mathbb{F}_q)^2$ and $q$ is a prime power. The construction method for such a transversal design is simple; we refer the reader to the relevant discussion in Section~\ref{subsec: transversal designs} on Theorem~\ref{linear OA thm} and Corollaries~\ref{linear OA p.p.thm}~and~\ref{linear TD p.p. thm}.

In particular, we can construct a linear $\TD_1(t,a,q)$ for a prime power $q \geq a$ by first constructing a $\TD_1(t,q,q)$ and then (if necessary) deleting $q-a$ groups.  This yields a wide range of informative and weakly $(t-c)$-secure $(a, n-a-c, c)$-strategies for card decks of size $n=aq$ and any choice of $c$ satisfying $c \leq \min \{t-1, a-t\}$. If we take $t = a-c$, these strategies are optimal. We summarize this result in the following theorem.

\begin{theorem} Consider the transversal Russian cards problem. Let $q$ be a prime power such that $q \geq a$ and $c \leq \frac{a-1}{2}$. Then there exists an equitable $(a,aq-a-c,c)$-strategy that is optimal,  informative for Bob, and weakly $(a-2c)$-secure against Cathy.
\end{theorem}

\section{Discussion and comparison with related work}
\label{sec: Related Work}

The Russian cards problem and variants of it has received a fair amount of attention in the literature, with focus ranging from possible applications to key generation~\cite{FPR91,FW91,FW93,FW93_2,FW96,MSN02,KMN04,ACDFJS11,ACDFJS11_ext}, to analyses based on epistemic logic\cite{D03,D05,DHMR06, CK08}, to card deals with more than three players~\cite{DY10,HD11}. Of more relevance to our work is the recent research that takes a combinatorial approach~\cite{AAADH05,ADR09,ACDFJS11_ext,ACDFJS11,CDFJS,cardsSwansonS12}, on which we now focus. 

Many useful results concerning parameter bounds and announcement sizes for weak 1-security, some of which we use in this paper, are given by Albert et al.~\cite{AAADH05}. Albert et al.~\cite{ACDFJS11,ACDFJS11_ext} and Cord\'{o}n-Franco et al.~\cite{CDFJS} discuss protocols for card deals of a particular form that achieve weak 1-security, using card sums  modulo an appropriate parameter for announcements. Atkinson et al.~\cite{ADR09} is the only work of which we are aware that treats security notions stronger than weak 1-security, other than work by Swanson and Stinson~\cite{cardsSwansonS12} and subsequent work by Cord\'{o}n-Franco et al.~\cite{CDFS13}.

In addition, there has been recent work~\cite{DS11,CDFS12} in which protocols consisting of more than one announcement by Alice and Bob are considered, which is a generalization of the problem which we consider here. Van Ditmarsch and Soler-Toscano~\cite{DS11} show that no good announcement exists for card deals of the form $(4,4,2)$ using bounds from Albert et al.~\cite{AAADH05}. The authors instead give an interactive protocol that requires at least three rounds of communication in order for  Alice and Bob to learn each other's hands; their protocol uses combinatorial designs to determine the initial announcement by Alice and the protocol analysis is done using epistemic logic. 

Cord\'{o}n-Franco et al.~\cite{CDFS12} consider four-step solutions that achieve weak $1$-security for the generalized Russian cards problem with parameters $(a, b, c)$ such that $c > a$; this is the first work that shows it is possible to achieve weak $1$-security in cases where Cathy holds more cards than one of the other players. The authors demonstrate the existence of a necessary construction for Bob's announcement when the card deal parameters satisfy specific conditions and briefly address the feasibility of finding such constructions in practice. In particular, the authors leave as an interesting open problem efficient algorithms for producing Bob's announcement.

In this paper, we build extensively on results by Swanson and Stinson~\cite{cardsSwansonS12}. In particular,  we greatly simplify the proofs for results connecting certain types of perfectly $\delta$-secure deals and Steiner systems, originally shown in Swanson and Stinson~\cite{cardsSwansonS12}. The construction technique using a ``starting design'', given in Theorem~\ref{aut thm gen} is a generalization of the technique given by Swanson and Stinson~\cite{cardsSwansonS12}. This generalized construction technique allows us to answer in the affirmative the question on the existence of perfectly secure and informative strategies for deals in which Cathy holds more than one card. 

Cord\'{o}n-Franco et al.~\cite{CDFS13} further elaborate on protocols of length two and the notion of weak $\delta$-security. The authors present a geometric protocol, discussed in Section~\ref{subsec: geom}, based on hyperplanes that yields informative and weakly $\delta$-secure equitable $(a,b,c)$-strategies for appropriate parameters. In particular, this protocol allows Cathy to hold more than one card. In certain card deals, this protocol achieves perfect $\delta$-security for $\delta$ equal to one or two. We remark that with the exception of Section~\ref{subsec: geom}, our results were completed independently of Cord\'{o}n-Franco et al.~\cite{CDFS13}.

\section{Concluding remarks and future work}
\label{sec: Conclusion}

We give a characterization for solutions to the generalized Russian cards problem that are perfectly $\delta$-secure. That is, we show an equivalence between a $\gamma$-equitable strategy that is perfectly $\delta$-secure for some $\delta$ and a set of $(c+\delta)$-designs on $n$ points with block size $a$, where this set must satisfy the additional property that every $a$-subset of $X$ occurs in precisely $\gamma$ of these designs.

Building on the results of Swanson and Stinson~\cite{cardsSwansonS12}, we show how to use a ``starting'' $t$-$(n,a,1)$-design to construct equitable $(a,b,c)$-strategies that are informative and perfectly $(t-c)$-secure against Cathy for any choice of $c$ satisfying $c \leq \min \{t-1,a-t\}$. In particular, this indicates that if an appropriate $t$-design exists, it is possible to achieve perfect security for deals where Cathy holds more than one card.  We present an example construction, based on inversive planes, for $(q+1, q^2-q-2, 2)$-strategies which are perfectly 1-secure against Cathy and informative for Bob, where $q$ is a prime power. We also analyze the security properties of Cord\'{o}n-Franco et al.'s~\cite{CDFS13} geometric protocol, remarking that this protocol yields a nice construction for a 3-design for certain parameters.

In addition, we discuss a variation of the Russian cards problem which admits nice solutions using transversal designs. The variant changes the manner in which the cards are dealt, but the resulting problem can be solved using large sets of transversal designs with $\lambda=1$ and arbitrary $t$, which are easy to construct. In particular, this solution is optimal in terms of the number of announcements and provides the strongest possible security for appropriate parameters. That is, for card decks of size $aq$, where $q \geq a$ is a prime power, we achieve $(a,aq-a-c,c)$-strategies that are optimal,  informative for Bob, and weakly $(a-2c)$-secure against Cathy for $c \leq \frac{a-1}{2}$.

There are many open problems in the area, especially for deals with $c > 1$. Given the general difficulty of constructing $t$-designs for $t > 2$ and $\lambda =1$, we see that constructing perfectly $\delta$-secure and informative strategies for $c >1$ is a difficult combinatorial problem. A more promising direction for the case $c >1$ may be strategies that are weakly $\delta$-secure for $\delta >1$, a concept first introduced by Swanson and Stinson~\cite{cardsSwansonS12}, which has received some attention in current literature~\cite{CDFS13}. In particular, further characterizing such strategies using combinatorial notions might prove informative.

\section*{Acknowledgments}
We would like to thank the anonymous referee for their valuable comments.

\bibliographystyle{splncs03}
\bibliography{cards}
\end{document}